%% file: main.tex
\documentclass[reprint,aps,prl,superscriptaddress,nofootinbib,longbibliography]{revtex4-1}
\usepackage{tcolorbox} 
\usepackage{xcolor}
\usepackage[utf8]{inputenc}
\usepackage{verbatim}
\usepackage{color}
\usepackage{times}
\usepackage{amsmath}
\usepackage{subcaption} 
\usepackage{amssymb}
\usepackage{stmaryrd}
\usepackage{makecell}
\usepackage{xr-hyper}
\usepackage{graphicx}
\usepackage{tikz-cd}
\usepackage{amsthm} 
\newtheorem{theorem}{Theorem}
\usepackage{adjustbox}
\usepackage[unicode=true,
bookmarks=true,bookmarksnumbered=false,bookmarksopen=false,
 breaklinks=false,pdfborder={0 0 1},backref=false,colorlinks=true]
 {hyperref}
\hypersetup{
 linkcolor=magenta, urlcolor=blue, citecolor=blue, pdfstartview={FitH}, hyperfootnotes=true, unicode=true}
\makeatletter

\def\be{\begin{equation}}
\def\ee{\end{equation}}

\definecolor{burntorange}{rgb}{0.8, 0.33, 0.0}

\newcommand{\doublewidetilde}[1]{{%
  \mathpalette\double@widetilde{#1}%
}}
\newcommand{\double@widetilde}[2]{%
  \sbox\z@{$\m@th#1\widetilde{#2}$}%
  \ht\z@=.9\ht\z@
  \widetilde{\box\z@}%
}

\usepackage{amsfonts}\usepackage{tabularx}\usepackage{dcolumn}\usepackage{bm}\usepackage{graphicx}\usepackage{epstopdf}

\hypersetup{urlcolor=blue}
\usepackage{xr}
\usepackage{lipsum}
\usepackage{tensor}
\usepackage{braket}

\usepackage[capitalise,compress]{cleveref}
\crefname{section}{Sec.}{Secs.}
\Crefname{section}{Section}{Sections}
\crefrangelabelformat{equation}{\textup{(#3#1#4)}--\textup{(#5#2#6)}}

\makeatother

\usepackage{xcolor}

\definecolor{darkgreen}{rgb}{0.0, 0.6, 0.13}

\definecolor{dkblue}{rgb}{0,0,0.5}

\usepackage{amsfonts}\usepackage{tabularx}\usepackage{dcolumn}\usepackage{bm}\usepackage{graphicx}\usepackage{epstopdf}

\hypersetup{urlcolor=blue}

\def\ket#1{\left|#1\right\rangle}

\def\braket#1{\left\langle#1\right\rangle}

\usepackage{xr}
\usepackage{xr-hyper}
\usepackage{hyperref}
\makeatletter
\newcommand*{\addFileDependency}[1]{
  \typeout{(#1)}
  \@addtofilelist{#1}
  \IfFileExists{#1}{}{\typeout{No file #1.}}
}
\makeatother

\makeatother
\begin{document}

\title{Unraveling Quantum Environments: Transformer-Assisted Learning in Lindblad Dynamics}

\author{Chi-Sheng Chen}
\affiliation{Neuro Industry Research, Neuro Industry, Inc., Boston, MA, USA}

\author{En-Jui Kuo}
\affiliation{Department of Electrophysics, National Yang Ming Chiao Tung University, Hsinchu, Taiwan, R.O.C.}

\begin{abstract}
Understanding dissipation in open quantum systems is crucial for the development of robust quantum technologies. In this work, we introduce a Transformer-based machine learning framework to infer time-dependent dissipation rates in quantum systems governed by the Lindblad master equation. Our approach uses time series of observable quantities, such as $\langle \sigma_x(t) \rangle$, $\langle \sigma_y(t) \rangle$, and $\langle \sigma_z(t) \rangle$, as input to learn dissipation profiles without requiring knowledge of the initial quantum state and even Hamiltonian.

We demonstrate the effectiveness of our approach on a hierarchy of open quantum models of increasing complexity, including single-qubit systems with time independent or time dependent jump rates, two-qubit interacting systems (e.g., Heisenberg and transverse Ising models), and the Jaynes–Cummings model involving light–matter interaction and cavity loss with time dependent decay rates. Our method accurately reconstructs both fixed and time-dependent decay rates from observable time series. To support this, we prove that under reasonable assumptions, the jump rates in all these models are uniquely determined by a finite set of observables, such as qubit and photon measurements. In practice, we combine Transformer-based architectures with lightweight feature extraction techniques to efficiently learn these dynamics. Our results suggest that modern machine learning tools can serve as scalable and data-driven alternatives for identifying unknown environments in open quantum systems.
\end{abstract}
\maketitle

\section{Introduction}

Open quantum systems play a fundamental role in quantum information, condensed matter physics, and quantum thermodynamics~\cite{breuer2016colloquium, rivas2012open}. Unlike closed systems, which evolve unitarily under the Schrödinger equation, open quantum systems interact with external environments, resulting in non-unitary dynamics. Environmental coupling introduces dissipation and decoherence, making the analysis of such systems essential for understanding real-world quantum technologies~\cite{sarandy2005adiabatic}.

In the weak coupling regime, the dynamics of an open system's density matrix $\rho$ are often modeled by the Lindblad master equation~\cite{breuer2016colloquium}:
\begin{align}\label{eq:lin}
    \dot{\rho} = -i[H,\rho] + \sum_{i} \gamma_{i} \left( L_{i} \rho L_{i}^{\dagger } - \frac{1}{2} \{L_{i}^{\dagger} L_{i}, \rho \} \right),
\end{align}
where $H$ governs the unitary evolution, $L_i$ are jump operators that encode system–environment interactions, and $\gamma_i \geq 0$ are decay rates that quantify the coupling strength.

In many experimental settings, both the jump operators and the dissipation rates $\gamma_i$ can be complex or time-dependent~\cite{banchi2018modelling, fujimoto2022impact, yin2014nonequilibrium, barontini2013controlling}. This uncertainty motivates a sensing-based approach: embed a known quantum system (e.g., a qubit or photonic state) into an unknown environment and infer the dissipation parameters from observable dynamics. This concept underlies quantum sensing protocols~\cite{degen2017quantum}, which have been demonstrated on platforms including superconducting flux qubits~\cite{bylander2011noise}, photonic systems~\cite{wang2020quantum}, and trapped ions~\cite{wolf2021quantum}.

In this work, we explore whether such observable time-series data can be used to infer environmental properties, and propose using modern sequence modeling architectures—specifically the Transformer~\cite{vaswani2017attention, li2019enhancing, zerveas2021transformer, lim2021temporal}—to learn time-dependent dissipation rates in open quantum systems.

We begin with the single-qubit case, where the jump operators are fixed but the dissipation rate $\gamma(t)$ is unknown. The model input consists of local observable trajectories: $\langle \sigma_x(t) \rangle$, $\langle \sigma_y(t) \rangle$, and $\langle \sigma_z(t) \rangle$. We demonstrate that a Transformer-based regression model can accurately reconstruct the time-dependent $\gamma(t)$ from these observables.

We then extend our analysis to two-qubit systems with interacting Hamiltonians. In particular, we consider the anisotropic Heisenberg model (realized in quantum dot arrays~\cite{qiao2020coherent, zhou2024exchange, hu2003overview, chan2021exchange, kandel2019coherent}) and the transverse-field Ising model (implemented in ion-trap and Rydberg platforms). In both cases, the Transformer model successfully predicts the dissipation rates associated with each qubit, showing strong generalization across coupling regimes.

To further test the method in hybrid quantum systems, we apply it to the open Jaynes–Cummings model~\cite{shore1993jaynes}, which describes light–matter interaction between a qubit and a photon mode. Here, the model predicts six time-dependent coefficients parameterizing the dissipation rates $\kappa(t)$ and $\gamma(t)$, based on the observable sequences $\langle \sigma_s(t) \rangle, s\in \{x,y,z \}$, $\langle a^\dagger a(t) \rangle$, and $\langle X \otimes \sigma_z(t) \rangle$.

In all cases, the Transformer model achieves high predictive accuracy. We summarize full architectural and training details in the Supplemental Material [\onlinecite{SM}, Sec.B-D]. Moreover, we provide theoretical results justifying the identifiability of dissipation rates from observable trajectories under mild assumptions. Analytical derivations for both the single-qubit Lindblad case and the Jaynes–Cummings model are presented in Supplemental Material [\onlinecite{SM}, Sec.~E].

Our results suggest that data-driven sequence models can serve as effective and scalable tools for extracting physical parameters in open quantum systems, offering a robust alternative to analytic inversion or system identification methods.

\subsection{MACHINE LEARNING METHOD BASICS}
 In recent years, machine learning has emerged as a powerful tool across various domains of physics \cite{carleo2019machine}. Applications include phase classification in condensed matter systems \cite{wetzel2017unsupervised,van2017learning, kuo2022decoding, kuo2022unsupervised}, modeling non-equilibrium dynamics \cite{van2018learning,schindler2017probing}, exploring the string theory landscape and holographic dualities \cite{carifio2017machine,PhysRevD.98.046019}, simulating quantum many-body systems \cite{carleo2017solving}, quantum state tomography \cite{torlai2018neural,carrasquilla2019reconstructing}, enhancing quantum hardware performance \cite{torlai2019integrating}, and quantum machine learning \cite{liu2024training, liu2024quantum}. In molecular physics and chemistry, machine learning has also been applied to accelerate molecular dynamics simulations \cite{tsai2020learning, tsai2022path}.
 
When specilize in time series prediction or Natural language processing, a wide array of architectures have been developed to capture temporal correlations, starting with Recurrent Neural Networks (RNNs) \cite{rumelhart1986learning}, Gated Recurrent Units (GRUs) \cite{cho2014learning}, and several gated variants like Long Short-Term Memory (LSTM) networks \cite{hochreiter1997long}. More recently, the Transformer architecture \cite{vaswani2017attention} has revolutionized the field by replacing recurrence with self-attention mechanisms, enabling scalable modeling of long-range dependencies. Its success in NLP has inspired applications across fields, including physics \cite{huang2023physicsformer,vaswani2021scaling,li2022physics}.

\section{Model and Quantities of Interest}\label{s:model}

We now describe the learning task and modeling setup. As discussed in Eq.~\eqref{eq:lin}, the dynamics of open quantum systems are governed by the Lindblad master equation~\cite{breuer2002theory, pearle2012simple}, where $H$ is the system Hamiltonian, $L_i$ are jump operators, and $\gamma_i(t) \geq 0$ are time-dependent dissipation rates. This form ensures the dynamics remain completely positive and trace preserving (CPTP)~\cite{nielsen2002quantum}. Conversely, any divisible CPTP map admits a representation of this structure.

We assume that both the Hamiltonian $H$ and jump operators $\{L_i\}$ are known, but the dissipation rates $\{\gamma_i(t)\}$ are unknown. Our goal is to learn these rates from observable time-series data. For example, consider a single-qubit system with Hamiltonian $H = \sigma_z$ and a single jump operator $\sigma_-$ with time-dependent dissipation $\gamma(t)$:
\begin{equation}
\dot{\rho} = -i[H,\rho] + \gamma(t) \left( \sigma_- \rho \sigma_+ - \frac{1}{2} \{ \sigma_+ \sigma_-, \rho \} \right).
\end{equation}
If the system is initialized in multiple states and we record observables such as $\langle \sigma_x(t) \rangle$, $\langle \sigma_y(t) \rangle$, and $\langle \sigma_z(t) \rangle$ over time, then in principle $\gamma(t)$ can be reconstructed, as the Pauli observables span the full space of Hermitian operators on a qubit. However, the exact inversion formulas are analytically involved; see Supplemental Material [\onlinecite{SM}, Sec.~E] for explicit derivations.

To proceed, we cast the reconstruction of $\gamma(t)$ as a regression problem. The input consists of observable trajectories generated from random initial states, and the output is a finite-dimensional parameterization of $\gamma(t)$. To represent the time-dependent dissipation profile, we adopt the Bernstein polynomial basis of fixed degree $n$. This basis offers a simple yet expressive representation for non-negative functions on $[0,1]$. Specifically, the Bernstein basis of degree $n$ consists of:
\begin{equation}
b_{j,n}(t) = \binom{n}{j} t^j (1 - t)^{n-j}, \quad j = 0, 1, \ldots, n,
\end{equation}
and any non-negative linear combination $\gamma(t) = \sum_{j=0}^{n} a_j b_{j,n}(t)$ with $a_j \geq 0$ guarantees $\gamma(t) \geq 0$ on $[0,1]$. This avoids the sign-indefiniteness inherent in monomial bases (e.g., $\{1, t, t^2\}$), and enables random sampling or learning over the coefficient space. A detailed discussion of the Bernstein basis and its approximation properties is provided in Supplemental Material [\onlinecite{SM}, Sec.~A].

The learning task is thus reduced to predicting the coefficient vector $\{a_0, \dots, a_n\}$ from time-series input features. In our implementation, we use a Transformer model trained on multiple trajectories to output these coefficients. The model is trained using a supervised loss over many initial states and dissipation profiles.

An overview of the entire pipeline is depicted in Fig.~\ref{fig:transformer_schematic}. Observable time series such as $\langle \sigma_x(t) \rangle$, $\langle \sigma_y(t) \rangle$, and $\langle \sigma_z(t) \rangle$ are encoded as fixed-length feature vectors. The Transformer learns to regress the dissipation profile by predicting the corresponding Bernstein coefficients. This framework will be applied in the next sections to various systems, including single-qubit models, two-qubit interactions, and the Jaynes–Cummings model.

\begin{figure}[t]
    \centering
    \includegraphics[width=0.48\textwidth]{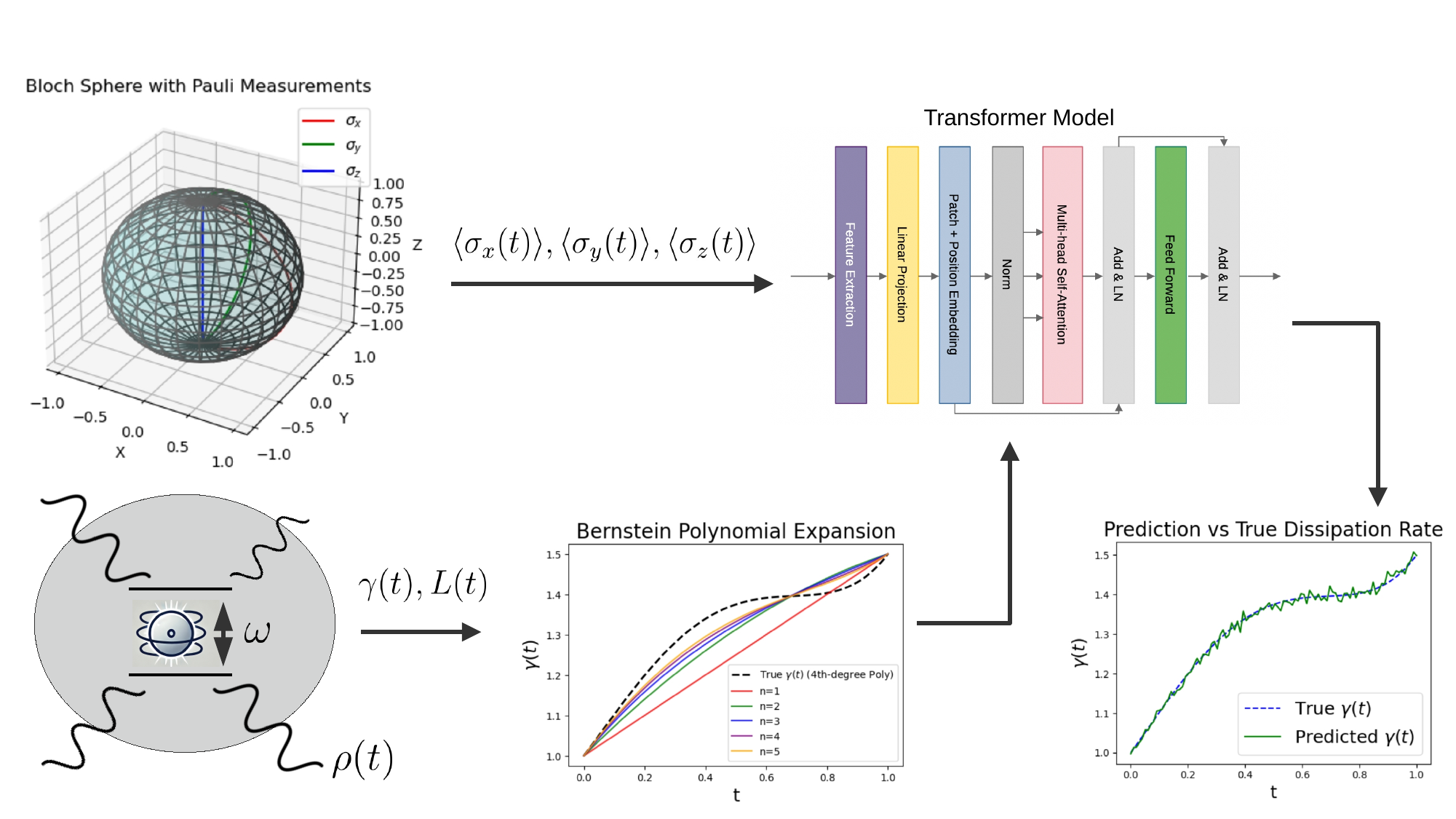}
    \caption{
    Schematic illustration of the machine learning pipeline. Time-series data of local observables $\langle \sigma_{x,y,z}(t) \rangle$ are used as input features to a Transformer model trained to regress the time-dependent dissipation rate $\gamma(t)$. The dissipation profile is parameterized using a Bernstein polynomial expansion. This approach is applied to single- and two-qubit systems, as well as light–matter interaction models.
    }
    \label{fig:transformer_schematic}
\end{figure}

\section{Examples}

In this section, we apply our framework to several concrete systems and summarize the setup of each in Table~\ref{table:ex}.

\subsection{Single-Qubit Dissipative Systems}

We begin with four representative examples involving single-qubit open dynamics. The first case considers a simple Hamiltonian $H = \sigma_z$ with a constant but unknown dissipation rate $\gamma_-$ associated with the jump operator $\sigma_-$. The goal is to predict the scalar value of $\gamma_-$ from the observable trajectory $\langle \sigma_z(t) \rangle$.

The second case extends this setup to time-dependent dissipation: $\gamma_-(t)$ is parameterized by a second-order Bernstein polynomial. In both the first and second examples, the system is initialized in the ground state $\ket{0}$.

In the third and fourth examples, the initial state is drawn randomly from the Haar measure over pure qubit states. This is achieved by sampling the components $\alpha_j$ from independent complex Gaussian distributions $\mathbb{C}\mathcal{N}(0,1)$ and normalizing the resulting vector. These random initializations provide a more general testing ground for the model's robustness across input variability.

In all four cases, we generate 1000 simulation trajectories using the QuTiP package~\cite{johansson2012qutip}, each with randomly sampled dissipation rates. The input features consist of observable time series, either $\langle \sigma_z(t) \rangle$ or $\langle \sigma_{x,y,z}(t) \rangle$ depending on the initial state. The model is trained to regress either the constant $\gamma$ or the Bernstein coefficients if $\gamma(t)$ is time-dependent. The dataset is split into 80\% training and 20\% testing.

Model architecture, feature extraction techniques, and additional training details are provided in Supplemental Material [\onlinecite{SM}, Sec.~B]. Results for the first three examples are shown in Fig.~\ref{fig:single-qubit}. The top panel corresponds to the first example (constant $\gamma_-$), the second and third panels illustrate the performance on time-dependent $\gamma_-(t)$ with deterministic and random initial states, respectively. The $R^2$ scores in all three cases exceed $0.99$, indicating excellent predictive performance.

The fourth example introduces two time-dependent dissipation channels: $\gamma_+(t)$ and $\gamma_-(t)$, corresponding to jump operators $\sigma_+$ and $\sigma_-$. Each is represented using a second-order Bernstein basis. This results in a six-dimensional output space. Due to space constraints, results for this case are included in the Supplemental Material [\onlinecite{SM}, Sec.~B].

\begin{table*}[hbt]
    \centering
    \begin{tabular}{|c|c|c|c|c|}
        \hline
        \textbf{Hilbert Space}& \textbf{Hamiltonian} & \textbf{Initial}  &
        \textbf{Jump operator $L$, rate $\gamma$} & \textbf{Input}
       \\ \hline
      $\mathbb{C}^2$ &  $\sigma_z$                    & $\ket{0}$ & ($\sigma_{-}, \gamma_{-}$)  & $\braket{\sigma_z(t)}$\\ \hline
   $\mathbb{C}^2$   & $\sigma_z$                 & Haar random & ($\sigma_{-}, \gamma_{-}$), ($\sigma_{+}, \gamma_{+}$) &  $\braket{\sigma_i(t)}, i \in \{x,y,z\}$  \\ \hline
   $\mathbb{C}^2$  & $\sigma_z$                 &  $\ket{0}$ & ($\sigma_{-}, \gamma_{-}(t)$) & $\braket{\sigma_z(t)}$\\ \hline
 $\mathbb{C}^2$     &  $\sigma_z$     & Haar random  & ($\sigma_{-}, \gamma_{-}(t)$), ($\sigma_{+}, \gamma_{+}(t)$)& $\braket{\sigma_i(t)}, i \in \{x,y,z\}$  \\ \hline
  $\mathbb{C}^2\otimes \mathbb{C}^2$    &  $J(\sigma_x^{(1)}\sigma_x^{(2)} +\sigma_y^{(1)}\sigma_y^{(2)} +\sigma_z^{(1)}\sigma_z^{(2)} )$     & Haar random  & ($\sigma_{-,j}, \gamma_{-,j}$), ($\sigma_{+,j}, \gamma_{+,j}$)& $\braket{\sigma_i^j(t)}, i \in \{x,y,z\}, j\in\{ 1,2\}$  \\ \hline
    $\mathbb{C}^2\otimes \mathbb{C}^2$    &  $J \sigma_z^{(1)}\sigma_z^{(2)} + h (\sigma_x^{(1)} + \sigma_x^{(2)})$     & Haar random  & ($\sigma_{-,j}, \gamma_{-,j}$), ($\sigma_{+,j}, \gamma_{+,j}$)& $\braket{\sigma_i^j(t)}, i \in \{x,y,z\}, j\in\{ 1,2\}$ \\ \hline
     $\mathbb{C}^2\otimes \mathcal{H}_{p}$   &  $\omega_c a^\dagger a + \frac{\omega_q}{2} \sigma_z + g (\sigma_+ a + \sigma_- a^\dagger)$     & Haar random $\otimes \ket{n}$  & ($\sigma_{-}, \gamma_{-}(t)$), ($a, \kappa(t)$) & $\langle \sigma_i(t) \rangle, i\in \{x,y,z \}, \braket{n(t)}, \braket{\sigma_z \otimes X_{photon}(t)}$\\ \hline
    \end{tabular}
    \caption{Details of examples shown in the single-qubit space $\mathbb{C}^2$. For the first and third examples, we pick $\gamma_+, \gamma_{-}$ uniformly at random in $(0,2)$. For the third and fourth examples, we pick $\gamma_{\pm}(t) = \sum_{j=0}^{2} \gamma_{\pm, j} b_{j,2}(t)$, where $b_{j,2}(t)$ is the $j$-th Bernstein polynomial of degree $2$, with all $\gamma_{\pm, j} \in (0.1, 2)$. We denote $\mathcal{H}_{p}$ to be the Hilbert space of single photon.}
    \label{table:ex}
\end{table*}
\begin{figure}[t]
    \centering
    \begin{subfigure}[b]{0.48\textwidth}
        \includegraphics[width=\textwidth]{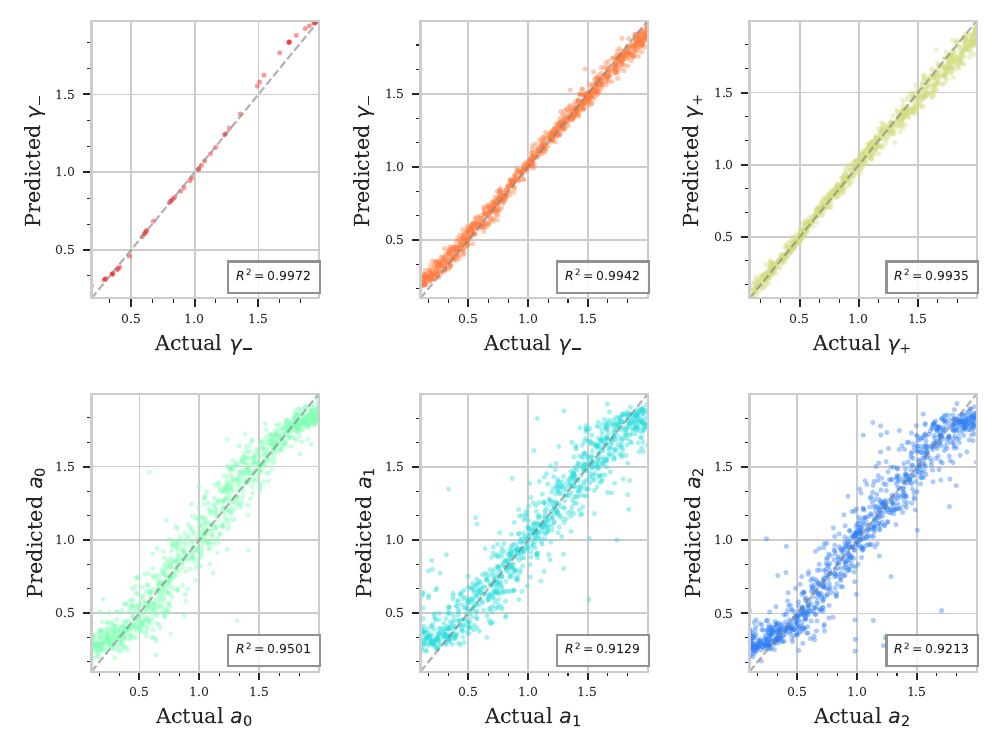}
        \caption{Single-qubit dissipation system}
        \label{fig:single-qubit}
    \end{subfigure}
    \hfill
    \begin{subfigure}[b]{0.48\textwidth}
        \includegraphics[width=\textwidth]{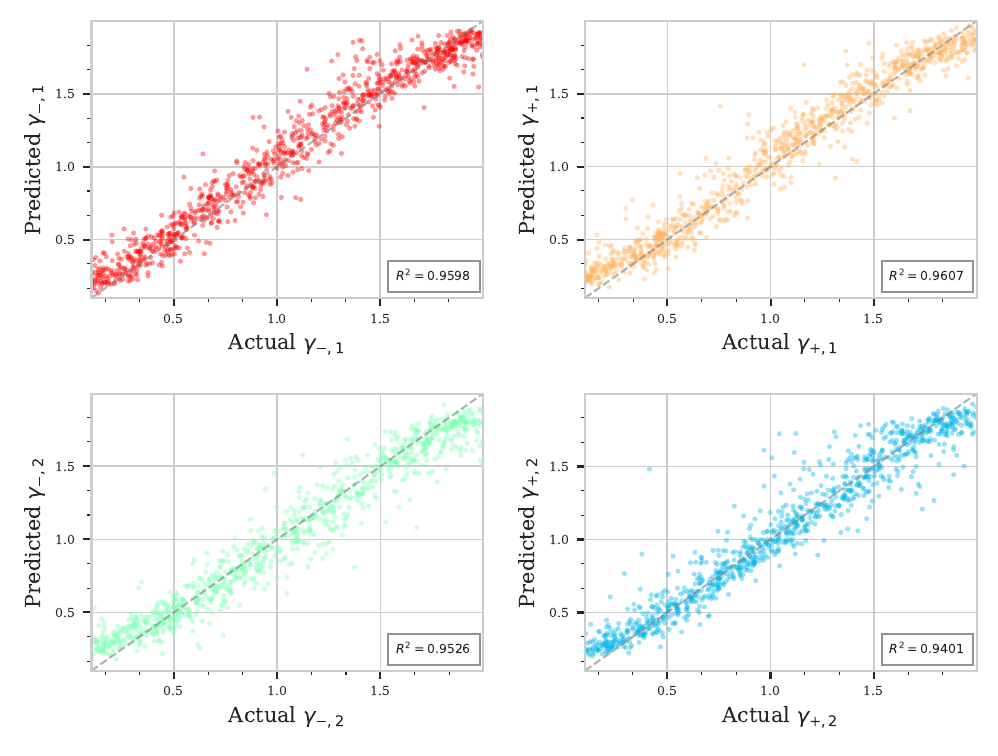}
        \caption{Two-qubit dissipation system (Heisenberg model)}
        \label{fig:two-qubit}
    \end{subfigure}
    \caption{
    Predictions of dissipation rates for single- and two-qubit systems. (a) shows results from the first three single-qubit examples in Table~\ref{table:ex}. The top row illustrates prediction of time-dependent dissipation rates. The first panel corresponds to Example~1, with a single channel $\gamma_-(t)$ and jump operator $\sigma_-$. The second and third panels correspond to Example~2, which includes both $\gamma_+(t)$ and $\gamma_-(t)$ associated with jump operators $\sigma_+$ and $\sigma_-$. The bottom row shows the predicted Bernstein coefficients $a_0$, $a_1$, and $a_2$ for Example~3, where $\gamma_-(t) = \sum_{k=0}^{2} a_k b_{k,2}(t)$ is expanded in the Bernstein basis. (b) shows the predicted dissipation rates for the Heisenberg two-qubit model, involving four constant rates $\gamma_{\pm,j}$ acting on each qubit $j = 1,2$.
    }
    \label{fig:combined-qubit-results}
\end{figure}

\subsection{Two-Qubit Dissipative Systems}

After validating our approach on single-qubit systems, we extend the framework to two-qubit dissipative systems with constant pumping and loss channels. We consider two widely studied Hamiltonians:

\begin{itemize}
    \item The Heisenberg exchange Hamiltonian: 
    \[
    H = J(\sigma_x^{(1)}\sigma_x^{(2)} + \sigma_y^{(1)}\sigma_y^{(2)} + \sigma_z^{(1)}\sigma_z^{(2)}),
    \]
    relevant to quantum dot architectures~\cite{qiao2020coherent, zhou2024exchange, hu2003overview, chan2021exchange, kandel2019coherent}.
    
    \item The transverse-field Ising model: 
    \[
    H = J \sigma_z^{(1)}\sigma_z^{(2)} + h(\sigma_x^{(1)} + \sigma_x^{(2)}),
    \]
    commonly realized in superconducting circuits, trapped ions, and cold-atom platforms~\cite{sachdev2011quantum, islam2011onset, zhang2017observation}.
\end{itemize}

Both models serve as minimal testbeds for studying dissipation, entanglement, and open-system dynamics under tunable interactions. For the Heisenberg model, we sample $J \in (0,2)$ uniformly. For the transverse-field Ising model, we sample $(J,h) \in (0.1,2)^2$ and assign four dissipation rates $\gamma_{\pm,j} \in (0.1,2)$ for $j=1,2$, corresponding to jump operators $\sigma_{\pm,j}$ acting on each qubit. 

The model input consists of local observables $\langle \sigma_{x,y,z}^{(1)}(t) \rangle$ and $\langle \sigma_{x,y,z}^{(2)}(t) \rangle$, and the regression target is the set of unknown rates $\{\gamma_{\pm,1}, \gamma_{\pm,2}\}$. Results for the Heisenberg model are shown in Fig.~\ref{fig:two-qubit}, where the Transformer accurately predicts all four rates with $R^2 \geq 0.95$.

Each trajectory is generated from a product of Haar-random pure states and evolved under the Lindblad master equation. We simulate 3000 examples for each model. The
dataset is split into 80 $\%$ training and 20 $\%$ testing. Feature extraction techniques are applied to improve performance, and full implementation details are provided in Supplemental Material [\onlinecite{SM}, Sec.~C].

These two-qubit models are summarized in columns 5 and 6 of Table~\ref{table:ex}. The dataset offers a testbed for learning dissipative dynamics with tunable interactions and multi-qubit decoherence. Theoretical justification for identifiability of the dissipation rates in this setting is given in Supplemental Material [\onlinecite{SM}, Sec.~E].

\subsection{Photon–Spin Coupling Systems}

We now apply our framework to hybrid light–matter systems, focusing on the Jaynes–Cummings model~\cite{shore1993jaynes}. The system is governed by the Hamiltonian
\begin{equation}
H = \omega_c a^\dagger a + \frac{\omega_q}{2} \sigma_z + g (\sigma_+ a + \sigma_- a^\dagger),
\end{equation}
where $\omega_c$ and $\omega_q$ denote the cavity and qubit frequencies, respectively, and $g$ is the light–matter coupling strength.

To reflect realistic parameter regimes, we sample $(\omega_c, \omega_q)$ uniformly from $(0.8, 1.2)$ and $g$ from $(0.01, 0.1)$. The system is initialized in a product state $\ket{n} \otimes \ket{g}$, where $1 \leq n \leq 10$ is drawn randomly. The open-system dynamics are governed by a time-dependent Lindblad master equation:
\begin{equation}
\dot{\rho} = -i [H, \rho] + \kappa(t)\, \mathcal{D}[a]\rho + \gamma(t)\, \mathcal{D}[\sigma_-]\rho,
\end{equation}
where $\kappa(t)$ and $\gamma(t)$ represent photon and qubit decay rates, respectively. Both are parameterized as second-order Bernstein polynomials:
\[
\kappa(t) = \sum_{j=0}^{2} a_{j,\kappa} b_{j,2}(t), \quad \gamma(t) = \sum_{j=0}^{2} a_{j,\gamma} b_{j,2}(t),
\]
with coefficients $a_{j,\kappa}, a_{j,\gamma}$ sampled uniformly from $(0.01, 0.2)$.

The model input consists of time-series observables $\langle \sigma_s(t) \rangle s\in\{x,y,z\}$,  $\langle a^\dagger a(t) \rangle$, and $\langle X \otimes \sigma_z(t) \rangle$, recorded at 100 points over $t \in [0,10]$. The physical parameters $(\omega_c, \omega_q, g)$ are excluded from the input. The regression target is the six Bernstein coefficients $\{a_{j,\gamma}, a_{j,\kappa}\}_{j=0}^{2}$ encoding the dissipation rates.

We generate 5000 training examples and train the model in a supervised setting with dataset splited into 80 $\%$ training and 20 $\%$ testing. Results are shown in Fig.~\ref{fig:open}, where all predicted coefficients achieve $R^2$ scores above 0.9. Full model specifications and training procedures are described in Supplemental Material [\onlinecite{SM}, Sec.~D]. The identifiability of $\gamma(t)$ and $\kappa(t)$ from these observables is formally justified in Sec.~E. The setup for this case is summarized in column 7 of Table~\ref{table:ex}.

\begin{figure}[t]
    \centering
    \includegraphics[width=\columnwidth]{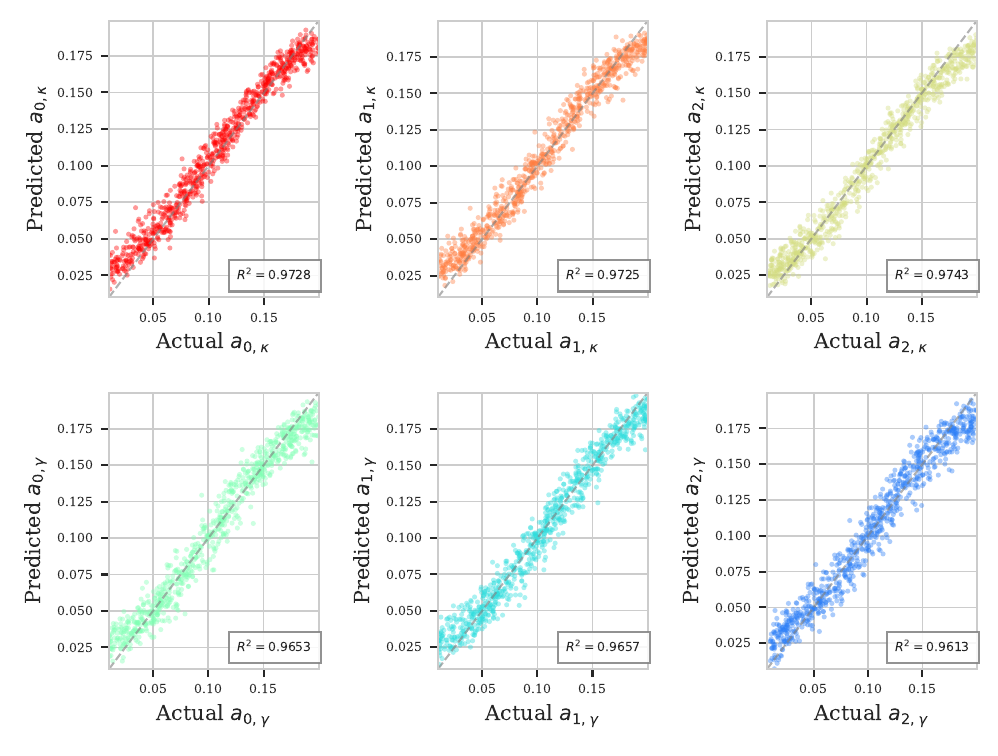}
    \caption{
    Prediction of time-dependent dissipation rates in the open Jaynes–Cummings model. The decay rates $\gamma(t)$ and $\kappa(t)$ are each parameterized by a second-order Bernstein polynomial with three unknown coefficients. The model accurately reconstructs all six coefficients from time series of observables $\langle \sigma_z(t) \rangle$, $\langle a^\dagger a(t) \rangle$, and $\langle X \otimes \sigma_z(t) \rangle$. See Table~\ref{table:ex}, column 7 for full input/output specification.
    }
    \label{fig:open}
\end{figure}

\section{Discussion} \label{s:discussion}

We have presented a Transformer-based framework for learning time-dependent dissipation rates in open quantum systems. By leveraging observable time series such as $\langle \sigma_{x}(t) \rangle$, $\langle \sigma_{y}(t) \rangle$, and $\langle \sigma_{z}(t) \rangle$, the model accurately reconstructs dissipation profiles across a range of quantum systems, from single- and two-qubit models to hybrid light–matter systems like the Jaynes–Cummings model. A key advantage of this approach is that it does not require knowledge of the system's initial quantum state, making it robust to initialization noise—a desirable property in experimental quantum platforms~\cite{georgescu2020divincenzo}.

While Bernstein polynomials were initially introduced to encode non-negative dissipation functions, their primary benefit lies in generating smooth and diverse training data rather than in improving predictive accuracy. This construction ensures consistency with physical constraints and facilitates efficient supervised learning.

Although our study focused on Lindblad dynamics under weak coupling, the proposed framework can extend to other open-system models. In particular, it is compatible with more general Markovian descriptions such as the Redfield equation~\cite{pollard1994solution}, which governs dynamics in systems with strong internal interactions. Under the secular approximation, the Redfield formalism reduces to the Lindblad form, making the learned mapping transferable. Beyond quantum systems, this methodology may also be applied to classical stochastic dynamics, including those described by Fokker–Planck equations.

In summary, our results demonstrate that modern sequence models—especially Transformers—offer a scalable, robust, and interpretable approach to learning dissipation in open quantum systems. This opens new avenues for data-driven identification of environment-induced effects across a broad range of quantum technologies.

\section{Acknowledgements}
EJK thanks National Yang Ming Chiao Tung University for its support. CSC thanks Neuro Industry, Inc. for its support.

\bibliography{references}

\clearpage
\onecolumngrid

\begin{center}
\textbf{\large Supplementary Material: Unraveling Quantum Environments: Transformer-Assisted Learning in Lindblad Dynamics}
\end{center}
\vspace{0.5cm}

\renewcommand{\thesection}{S\arabic{section}}  
\renewcommand{\theequation}{S\arabic{equation}}  
\setcounter{section}{0}
\setcounter{equation}{0}

\input{supp.tex}

\end{document}

%% file: supp.tex


\section*{Section A: Bernstein Polynomial Approximation}
\label{app:bernstein}

Let $ \{b_{\nu,n}(t)\}_{\nu=0}^n $ denote the Bernstein basis polynomials of degree $ n $, defined on $ [0,1]$ by
\begin{equation}
b_{\nu,n}(t) = \binom{n}{\nu} t^\nu (1 - t)^{n - \nu}, \quad \nu = 0, 1, \dots, n.
\end{equation}
These polynomials form a non-negative partition of unity and span the space$ \Pi_n $ of real polynomials of degree at most $ n $. Given a continuous function $ f: [0,1] \to \mathbb{R} $, the Bernstein approximation of $ f $ is
\begin{equation}
B_n[f](t) = \sum_{\nu=0}^n f\left( \frac{\nu}{n} \right) b_{\nu,n}(t).
\end{equation}
By the classical Weierstrass approximation theorem, $ B_n[f] \to f $ uniformly as $ n \to \infty $ for any continuous $ f $~\cite{stone1948generalized}.

Compared to Taylor expansions, Bernstein polynomials have two advantages: (i) they preserve non-negativity if all coefficients are non-negative, and (ii) they are globally defined over the entire interval. For example, the quadratic polynomial $ f(t) = t^2 - 2t + 2 $ is strictly positive on $ [0,1] $, but has negative Taylor coefficients. In contrast, any non-negative quadratic function on $ [0,1] $ can be expressed as a positive linear combination of the Bernstein basis $ \{(1-t)^2, t(1-t), t^2\} $.

This positivity makes the Bernstein basis particularly useful for generating random non-negative decay functions $ \gamma(t) $ in open quantum systems. In our work, we represent unknown time-dependent dissipation rates using degree-2 Bernstein polynomials:
\begin{equation}
\gamma(t) = \sum_{\nu=0}^2 a_\nu b_{\nu,2}(t), \quad a_\nu \geq 0.
\end{equation}
We show in Section BCD that these coefficients can be effectively learned from time-series observables via Transformer regression. The robustness of this parameterization is demonstrated across all examples.

\section*{Section B: Single-Qubit Model Details and Additional Example with Two Time-Dependent $\gamma_{\pm}(t)$}

\subsection{Training Architecture and Feature Design}

Our core model is based on the Transformer encoder architecture~\cite{vaswani2017attention}, which learns temporal dependencies via self-attention. Given an input matrix $X \in \mathbb{R}^{n \times d}$, where $n$ is the sequence length and $d$ is the feature dimension, the self-attention mechanism computes:
\begin{equation}
\mathrm{Attention}(Q, K, V) = \mathrm{softmax}\left( \frac{QK^\top}{\sqrt{d_k}} \right) V,
\end{equation}
where the query, key, and value matrices are defined as:
\[
Q = XW^Q, \quad K = XW^K, \quad V = XW^V,
\]
with $W^Q, W^K, W^V \in \mathbb{R}^{d \times d_k}$ being learned projection matrices. The dot product $QK^\top$ measures similarity between time steps, and the softmax weights modulate which values $V$ to attend to at each position.

For richer representation learning, we use multi-head attention:
\begin{align}
\mathrm{MultiHead}(X) &= \mathrm{Concat}(\mathrm{head}_1, \dots, \mathrm{head}_h) W^O, \\
\mathrm{head}_i &= \mathrm{Attention}(XW_i^Q, XW_i^K, XW_i^V),
\end{align}
where $W^O \in \mathbb{R}^{hd_k \times d}$ is an output projection, and $h$ is the number of attention heads (we use $h=8$). This allows the model to jointly learn from multiple subspaces of the input.

Each encoder layer consists of multi-head self-attention, followed by a feedforward block, layer normalization, and residual connections. The final representation is passed to a multilayer perceptron (MLP) to predict the target coefficients.

Our regression pipeline is based on a Transformer encoder model designed to predict dissipation profiles from Pauli observable trajectories. Each single-qubit dataset consists of three input time series $\sigma_x(t)$, $\sigma_y(t)$, and $\sigma_z(t)$, evaluated at 100 uniform time points over $[0,10]$.

We extract 18 handcrafted features per observable, including:
\begin{itemize}
    \item Basic statistics: mean, std, max, min, median, 25th/75th percentiles
    \item Temporal variation: range, IQR, slope, mean/std of $\Delta x$ and $\Delta^2 x$
    \item Distributional shape: skewness and kurtosis
    \item Spectral: max and average FFT magnitude (excluding DC component)
\end{itemize}
These yield a 54-dimensional feature vector per sample. All features and targets are standardized using $z$-score normalization. The model architecture includes:
\begin{itemize}
    \item Linear projection $\rightarrow$ LayerNorm $\rightarrow$ ReLU $\rightarrow$ Dropout
    \item 4-layer Transformer encoder with 8 attention heads, feedforward dimension 512
    \item MLP head with layer sizes 256–128–64 and linear output layer
\end{itemize}
We train using AdamW optimizer (lr=$10^{-3}$, weight decay=0.01), with MSE loss, ReduceLROnPlateau scheduler (patience 20), and early stopping (patience 30). Each model is trained up to 300 epochs with batch size 32. Final predictions are inverse-transformed, and evaluated using $R^2$ and MSE on the test set.

\subsection{Example-Specific Configurations (Examples 1–4 in the Main Text)}

\paragraph{Example 1: Constant $\gamma$} A minimal Transformer (2 encoder layers, 4 heads) predicts a single scalar $\gamma$ from the time-averaged $\sigma_{x,y,z}$. This baseline achieves $R^2 \geq 0.99$.

\paragraph{Example 2: Constant $\gamma_{\pm}$} Predicts two values $\gamma_{+}, \gamma_{-}$ using 30 input features (10 per Pauli observable). A 3-layer Transformer encoder with 8 heads is used. The model learns both rates jointly with high accuracy.

\paragraph{Examples 3–4: Time-Dependent $\gamma(t)$} We represent each dissipation rate using degree-2 Bernstein polynomial:
\[
\gamma(t) = \sum_{j=0}^2 a_j b_{j,2}(t), \quad b_{j,2}(t) = \binom{2}{j} t^j (1-t)^{2-j}.
\]
The model predicts the coefficient vector $(a_0, a_1, a_2)$ per channel. In Example 3, only one time-dependent $\gamma(t)$ is present; in Example 4, both $\gamma_{+}(t)$ and $\gamma_{-}(t)$ are time-dependent, requiring six total coefficients. Both models use the same 4-layer Transformer encoder as described above.

\subsection{Additional Results for Example 4 (Two Time-Dependent Dissipation Channels)}

We provide a detailed prediction plot for Example 4 (see Table in Main Text) in Fig.~\ref{fig:app:current-decay}. This task involves learning six Bernstein coefficients from $\sigma_{x,y,z}$ time series. The Transformer accurately recovers all coefficients, demonstrating strong generalization in the time-dependent multi-channel setting.

\begin{figure*}[t]
\centering
\includegraphics[width=\textwidth]{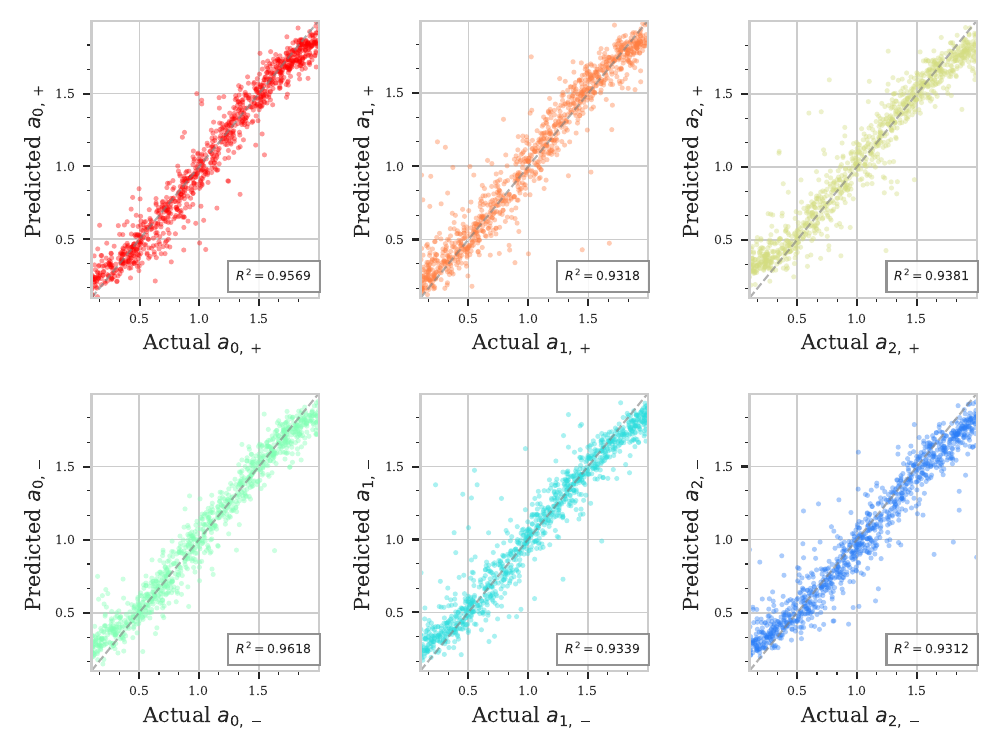}
\caption{
Prediction results for Example 4 (Main Text) involving a single-qubit system with two quadratic time-dependent loss channels. The dissipation rates are represented by second-order Bernstein polynomials: $\gamma_{\pm}(t) = \sum_{j=0}^{2} a_{j,\pm} b_{j,2}(t)$, and six coefficients are predicted from local observables.
}
\label{fig:app:current-decay}
\end{figure*}

\section*{Section C: Two-Qubit Models with Constant Dissipation Rates} \label{app:twoqubit}

This section provides the training details for both two-qubit models: the Heisenberg model presented in the main text, and the transverse-field Ising model included here as additional results (see Fig.~\ref{fig:Iontrap}). While their Hamiltonians differ, the dissipation structure and model training pipeline remain identical.

\paragraph{Heisenberg model (Main Text).} The Hamiltonian is given by
\begin{equation}
H = J(\sigma_x^{(1)} \sigma_x^{(2)} + \sigma_y^{(1)} \sigma_y^{(2)} + \sigma_z^{(1)} \sigma_z^{(2)}),
\end{equation}
where $J \sim \mathrm{Uniform}(0, 2)$ is the exchange coupling strength. Each qubit is coupled to a local environment with constant dissipation and gain rates $\gamma_{\pm,j}$ for $j=1,2$.

\paragraph{Transverse-field Ising model (Supplemental Result).} The Hamiltonian reads
\begin{equation}
H = J \sigma_z^{(1)} \sigma_z^{(2)} + h (\sigma_x^{(1)} + \sigma_x^{(2)}),
\end{equation}
with $(J,h) \sim \mathrm{Uniform}(0.1, 2)^2$. The jump operators and dissipation structure are the same as in the Heisenberg case.

\subsection*{Training Method (Shared by Both Models)}

For both models, we simulate 3000 trajectories. Each initial state is prepared as a tensor product of two Haar-random pure states. The systems evolve under the Lindblad equation with fixed dissipation rates $\gamma_{\pm,j}$. Time series of $\langle \sigma_{x,y,z}^{(j)}(t) \rangle$ for $j = 1,2$ are recorded at 100 time steps over $t \in [0, 1]$.

From each observable component, 10 statistical features are extracted—mean, standard deviation, extrema, percentiles, slope, and interquartile range—resulting in 60-dimensional input vectors. These vectors are standardized and projected into a 128-dimensional embedding space. We use a 4-layer Transformer encoder (8 heads, feedforward size 512), followed by an MLP decoder that predicts the four dissipation rates.

Training employs the AdamW optimizer ($\text{lr} = 10^{-3}$, weight decay 0.01), with ReduceLROnPlateau scheduling and early stopping (patience 30). Batch size is 32, and the model is trained for up to 300 epochs. The final models achieve $R^2 \geq 0.95$ on all outputs for both systems.

\begin{figure*}[t]
\centering
\includegraphics[width=\textwidth]{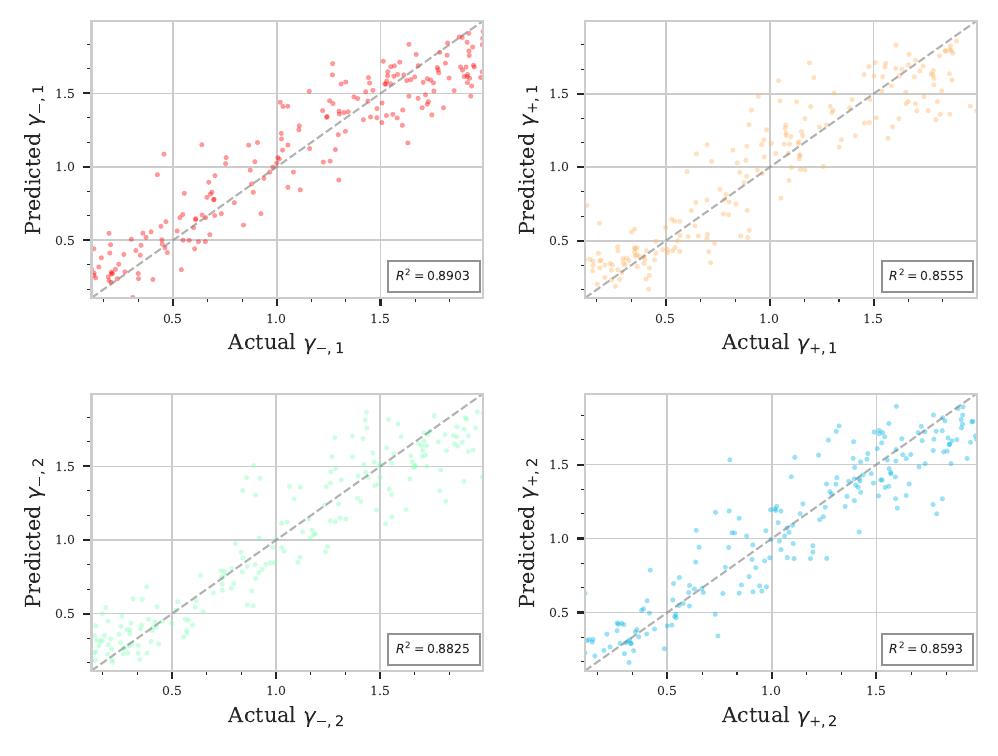}
\caption{
Prediction results for the transverse-field Ising model with four constant dissipation rates $\gamma_{\pm,j}$ ($j = 1,2$). Each dissipation rate is inferred from local observables $\langle \sigma_{x,y,z}^{(j)}(t) \rangle$ using a Transformer-based regression model. This confirms the method's generalizability beyond the Heisenberg model shown in the main text.
}
\label{fig:Iontrap}
\end{figure*}

\section*{Section D: Training Details for the Photon–Spin Coupling System} \label{app:photonspin}

This section describes the training procedure for the Transformer model used in the Jaynes–Cummings light–matter interaction setting. The goal is to predict six coefficients 
$(a_{0,\kappa}, a_{1,\kappa}, a_{2,\kappa}, a_{0,\gamma}, a_{1,\gamma}, a_{2,\gamma})$ that parameterize the time-dependent decay functions $\kappa(t)$ and $\gamma(t)$ as quadratic Bernstein polynomials.

Each data sample contains five observables recorded over 100 time steps: 
$\langle \sigma_x(t) \rangle$, $\langle \sigma_y(t) \rangle$, $\langle \sigma_z(t) \rangle$, $\langle a^\dagger a(t) \rangle$, and 
$\langle X_{\mathrm{photon}} \otimes \sigma_z(t) \rangle$.
From each signal, 18 time-series features are extracted—including statistical descriptors (mean, std, min, max, median, percentiles), shape (skewness, kurtosis), variation (range, IQR), temporal differences (mean/std of $\Delta x$, $\Delta^2 x$), slope, and spectral features (mean/max non-DC DFT amplitude). This results in a 90-dimensional input vector per sample.

All inputs and targets are standardized using $z$-score normalization. The 90-dimensional feature vector is projected into a 128-dimensional latent space via a feedforward block (Linear $\rightarrow$ LayerNorm $\rightarrow$ ReLU $\rightarrow$ Dropout). This is followed by a 4-layer Transformer encoder (8 heads per layer, feedforward dimension 512). The encoded representation is passed through an MLP with layers of size 256, 128, and 64, followed by a final output layer producing six regression values.

Training uses the AdamW optimizer (initial learning rate $10^{-3}$, weight decay 0.01) and mean squared error (MSE) loss. A ReduceLROnPlateau scheduler lowers the learning rate by a factor of 0.5 when validation loss stalls for 20 epochs. Early stopping is triggered after 30 epochs without improvement. We train with a batch size of 64 for up to 300 epochs. Gradient clipping (max norm 1.0) is applied for numerical stability.

Model performance is evaluated using the $R^2$ score and MSE on both training and held-out test sets. Predictions are inverse-transformed to physical units, and scatter plots are generated for each output coefficient. All models, logs, and output results are stored in timestamped directories for full reproducibility.

\section*{Section E: Learning Identifiability in Lindbladian Dynamics} \label{app:le}

We begin with a basic identifiability result showing that the time-dependent decay rate $\gamma(t)$ can be uniquely recovered from local observables in a single-qubit open quantum system governed by the Lindblad equation.

\begin{theorem}
Consider a single qubit in Hilbert space $\mathbb{C}^2$ with a time-independent Hamiltonian $H = \sigma_z$ and a single jump operator $L = \sigma_-$. Then the dissipation rate $\gamma(t)$ is uniquely determined by the observable $\langle \sigma_z(t) \rangle$ alone, without knowledge of the initial state.
\end{theorem}

\begin{proof}
The Lindblad equation is given by:
\[
\dot{\rho} = -i[\sigma_z, \rho] + \gamma(t)\left( \sigma_- \rho \sigma_+ - \frac{1}{2} \{\sigma_+ \sigma_-, \rho\} \right).
\]
We parametrize the density matrix as:
\[
\rho(t) = \begin{bmatrix} p_0(t) & \rho_{01}(t) \\ \rho_{01}^*(t) & p_1(t) \end{bmatrix},
\quad p_0(t) + p_1(t) = 1.
\]

Compute each term in the equation:

- The unitary part:
\[
-i[\sigma_z, \rho] = 
\begin{bmatrix}
0 & -2i\rho_{01} \\
2i\rho_{01}^* & 0
\end{bmatrix}.
\]

- The dissipative part:
\[
\sigma_- \rho \sigma_+ = 
\begin{bmatrix}
0 & 0 \\
p_0 & 0
\end{bmatrix},
\quad
\frac{1}{2} \{\sigma_+ \sigma_-, \rho\} = 
\begin{bmatrix}
p_0 & \rho_{01} \\
\rho_{01}^* & p_1
\end{bmatrix},
\]
so the dissipator becomes:
\[
\mathcal{D}[\sigma_-](\rho) = 
\begin{bmatrix}
-p_0 & -\rho_{01} \\
p_0 - \rho_{01}^* & -p_1
\end{bmatrix}.
\]

Combining both terms, we obtain the differential equations:
\begin{align}
\dot{p}_0 &= -\gamma(t) p_0, \\
\dot{p}_1 &= \gamma(t) p_0, \\
\dot{\rho}_{01} &= -2i\rho_{01} - \gamma(t) \rho_{01}.
\end{align}

Now, $\langle \sigma_z \rangle = p_0 - p_1 = 2p_0 - 1$, hence:
\[
p_0(t) = \frac{1 + \langle \sigma_z(t) \rangle}{2}.
\]

Solving $\dot{p}_0 = -\gamma(t)p_0$ gives:
\[
p_0(t) = p_0(0) e^{ - \int_0^t \gamma(s)\,ds }.
\]
Therefore,
\[
\langle \sigma_z(t) \rangle = -1 + 2p_0(0) e^{ - \int_0^t \gamma(s)\,ds }.
\]

Differentiating:
\[
\frac{d}{dt} \langle \sigma_z(t) \rangle = -2p_0(0) \gamma(t) e^{ - \int_0^t \gamma(s)\,ds }.
\]

Rewriting:
\[
\boxed{ \gamma(t) = - \frac{1}{1 + \langle \sigma_z(t) \rangle} \frac{d}{dt} \langle \sigma_z(t) \rangle. }
\]

This expression shows that $\gamma(t)$ can be reconstructed from $\langle \sigma_z(t) \rangle$ alone, regardless of the initial state. This validates the ability of our model to infer $\gamma(t)$ directly from observable trajectories.
\end{proof}
\begin{theorem}
Consider a single qubit in $\mathbb{C}^2$ with time-independent Hamiltonian $H = \sigma_z$, and jump operators $L_1 = \sigma_{-}$ and $L_2 = \sigma_{+}$ associated with time-dependent dissipation rates $\gamma_{-}(t)$ and $\gamma_{+}(t)$, respectively. Then, the functions $\gamma_{+}(t)$ and $\gamma_{-}(t)$ can be uniquely determined from the observables $\langle \sigma_x(t) \rangle$, $\langle \sigma_y(t) \rangle$, and $\langle \sigma_z(t) \rangle$.
\end{theorem}

\begin{proof}
Following the single-jump case, we now consider a more general dissipative setting with two jump channels. The master equation becomes:
\begin{align}
\dot{\rho} &= -i[\sigma_z, \rho] 
+ \gamma_{-}(t) \left( \sigma_- \rho \sigma_+ - \frac{1}{2} \{ \sigma_+ \sigma_-, \rho \} \right) \nonumber \\
&\quad + \gamma_{+}(t) \left( \sigma_+ \rho \sigma_- - \frac{1}{2} \{ \sigma_- \sigma_+, \rho \} \right).
\end{align}

Let $\rho(t)$ be written in the standard basis:
\[
\rho(t) = \begin{bmatrix} p_0(t) & \rho_{01}(t) \\ \rho_{01}^{*}(t) & p_1(t) \end{bmatrix}, \quad p_0(t) + p_1(t) = 1.
\]

Computing each Lindblad term yields:
\begin{align}
\sigma_- \rho \sigma_+ &= \begin{bmatrix} 0 & 0 \\ p_0 & 0 \end{bmatrix}, &
\sigma_+ \rho \sigma_- &= \begin{bmatrix} 0 & p_1 \\ 0 & 0 \end{bmatrix}, \\
\{\sigma_+ \sigma_-, \rho\}/2 &= \begin{bmatrix} p_0 & \rho_{01} \\ \rho_{01}^{*} & p_1 \end{bmatrix}, &
\{\sigma_- \sigma_+, \rho\}/2 &= \begin{bmatrix} p_0 & \rho_{01} \\ \rho_{01}^{*} & p_1 \end{bmatrix}.
\end{align}

Therefore, the two dissipative contributions are:
\begin{align}
\gamma_{-}(t) \left( \cdots \right) &= \gamma_{-}(t) \begin{bmatrix} -p_0 & -\rho_{01} \\ p_0 - \rho_{01}^* & -p_1 \end{bmatrix}, \\
\gamma_{+}(t) \left( \cdots \right) &= \gamma_{+}(t) \begin{bmatrix} -p_0 & p_1 - \rho_{01} \\ -\rho_{01}^* & -p_1 \end{bmatrix}.
\end{align}

Combining all terms, the equations of motion for each element of $\rho$ are:
\begin{align}
\dot{p}_0 &= -\gamma_{-}(t)\, p_0 + \gamma_{+}(t)\, p_1, \\
\dot{p}_1 &= -\dot{p}_0, \\
\dot{\rho}_{01} &= -2i \rho_{01} - \frac{1}{2}(\gamma_{-}(t) + \gamma_{+}(t)) \rho_{01}.
\end{align}

The Pauli expectation values are:
\begin{align}
\langle \sigma_x \rangle &= 2 \operatorname{Re}(\rho_{01}), &
\langle \sigma_y \rangle &= -2 \operatorname{Im}(\rho_{01}), \\
\langle \sigma_z \rangle &= p_0 - p_1 = 2p_0 - 1.
\end{align}
Define the complex observable:
\[
\frac{\langle \sigma_x \rangle - i \langle \sigma_y \rangle}{2} = \rho_{01}.
\]
Taking time derivatives, we obtain two equations involving $\gamma_{\pm}(t)$:
\begin{align}
\frac{1}{2} \frac{d}{dt} \langle \sigma_z(t) \rangle &= \gamma_{+}(t) - \gamma_{-}(t)\cdot \frac{1 + \langle \sigma_z(t) \rangle}{2}, \\
\frac{d}{dt} \left( \frac{\langle \sigma_x \rangle - i \langle \sigma_y \rangle}{2} \right)
&= \left(-2i - \frac{1}{2}(\gamma_{+}(t) + \gamma_{-}(t)) \right) \cdot \rho_{01}.
\end{align}

Solving these equations algebraically yields explicit expressions for $\gamma_{+}(t)$ and $\gamma_{-}(t)$ in terms of the observables and their time derivatives.
We denote time derivatives as $\dot{\langle O \rangle} := \frac{d}{dt} \langle O(t) \rangle$ for any observable $O$. One can solve for $\gamma_{-}(t)$ as:
\begin{equation}
\gamma_{-}(t) = \frac{
4 \langle \sigma_x \rangle \langle \sigma_z \rangle 
- i \langle \sigma_x \rangle \dot{\langle \sigma_z \rangle}
- 4i \langle \sigma_y \rangle \langle \sigma_z \rangle
- \langle \sigma_y \rangle \dot{\langle \sigma_z \rangle}
- 2i \langle \sigma_z \rangle \dot{\langle \sigma_x \rangle}
- 2 \langle \sigma_z \rangle \dot{\langle \sigma_y \rangle}
- 6i \dot{\langle \sigma_x \rangle}
- 6 \dot{\langle \sigma_y \rangle}
}{
2i \langle \sigma_x \rangle \langle \sigma_z \rangle
+ 4i \langle \sigma_x \rangle
+ 2 \langle \sigma_y \rangle \langle \sigma_z \rangle
+ 4 \langle \sigma_y \rangle
}.
\end{equation}

Once $\gamma_{-}(t)$ is obtained, $\gamma_{+}(t)$ can be recovered by substitution into the differential equations above. These expressions are exact and can be validated numerically against simulation data generated using QuTiP.

\end{proof}

We now generalize our theoretical analysis to multi-qubit open quantum systems. While exact analytical formulas such as those in the single-qubit case become intractable, we can still express the system evolution in a linear operator basis and outline the principle of identifiability. This justifies the feasibility of using machine learning to approximate the inversion from observable trajectories to underlying dissipation rates.

\subsection*{Pauli Basis Representation}

Let $\rho$ be the $n$-qubit density matrix, which can be expanded in the Pauli basis as
\[
\rho(t) = \sum_{\mu \in G_n} a_\mu(t) \mu,
\]
where $G_n$ denotes the $n$-qubit Pauli group, and $\{a_\mu(t)\}$ are real coefficients due to Hermiticity of $\rho$. The set $G_n$ spans a basis of dimension $4^n$, and the orthogonality relation $\mathrm{Tr}(\mu \nu) = \delta_{\mu,\nu}$ holds for $\mu,\nu \in G_n$.

Assume the Hamiltonian and Lindblad operators admit the following Pauli expansions:
\[
H = \sum_{\mu \in G_n} h_\mu \mu, \quad L_j = \sum_{\mu \in G_n} L_j^\mu \mu.
\]

We define the structure constants of the Pauli algebra via:
\[
[\mu, \nu] = \sum_{\lambda \in G_n} C^{\lambda}_{\mu,\nu} \lambda.
\]

\subsection*{Time Evolution of Coefficients}

Substituting the Pauli expansions into the Lindblad master equation, the derivative of the density matrix is:
\begin{align}
\sum_{\xi} \dot{a}_\xi \, \xi = & -i \sum_{\mu,\nu,\xi} h_\mu a_\nu C^{\xi}_{\mu,\nu} \, \xi + \sum_j \gamma_j \sum_{\mu,\nu,\lambda,\omega,\xi} \Bigg[
L_j^\mu a_\nu (L_j^\lambda)^* 
- \frac{1}{2}\Big( (L_j^\mu)^* a_\lambda L_j^\nu + (L_j^\nu)^* a_\mu L_j^\lambda \Big)
\Bigg] C^{\omega}_{\mu,\nu} C^{\xi}_{\omega,\lambda} \, \xi.
\end{align}
Here $\dot{a}_\xi$ represents the time derivative $\frac{da_\xi(t)}{dt}$.
By orthogonality, this yields a linear system for the coefficients:
\begin{align}
\dot{a}_\xi = & -i \sum_{\mu,\nu} h_\mu C^{\xi}_{\mu,\nu} a_\nu + \sum_j \gamma_j \sum_{\mu,\nu,\lambda,\omega}
\Bigg[
L_j^\mu a_\nu (L_j^\lambda)^* 
- \frac{1}{2} \Big( (L_j^\mu)^* a_\lambda L_j^\nu + (L_j^\nu)^* a_\mu L_j^\lambda \Big)
\Bigg] C^{\omega}_{\mu,\nu} C^{\xi}_{\omega,\lambda}.
\end{align}

This equation expresses $\dot{a}_\xi$ as a linear combination of $\gamma_j$ with coefficients that depend polynomially on $\{a_\mu(t)\}_{\mu \in G_n}$. Therefore, if the set $\{L_j\}$ is linearly independent as operators on $\mathbb{C}^{2^n}$, then for any $\xi$ the corresponding differential equations can be inverted to recover $\gamma_j(t)$ in principle.

\subsection*{Discussion and Practical Implications}

Although this linear system is exact, the inversion is prohibitively complex even for moderate $n$, involving $O(4^n)$ terms and tensor contractions over multiple indices. Moreover, estimating $\{a_\mu(t)\}$ from experiment requires full quantum state tomography. The sample complexity for estimating a rank-$r$ $n$-qubit state to Frobenius error $\epsilon$ is $O(2^n/\epsilon)$ \cite{o2016efficient}. Several efficient numerical methods exist, including Riemannian gradient descent over the quantum state manifold \cite{hsu2024quantum}.

This analysis supports our learning framework: instead of symbolically inverting this high-dimensional linear system, we use time-series observables as proxy input and train a Transformer to approximate the underlying mapping to $\gamma_j(t)$. The supervised data is generated by simulating the Lindblad evolution using tools such as QuTiP \cite{johansson2012qutip}.

In principle, a universal estimator can learn the function class $\{\gamma_j(t)\}$ through such indirect supervision, even when exact formulas are intractable.

\subsection{Photonic system learning ability}
\begin{theorem}
\label{thm:jc-identifiability}
Consider the open Jaynes--Cummings model with master equation
\begin{equation}
\dot{\rho}(t) = -i [H, \rho(t)] + \kappa(t)\, \mathcal{D}[a] \rho(t) + \gamma(t)\, \mathcal{D}[\sigma_-] \rho(t),
\end{equation}
where $H = \omega_c a^\dagger a + \frac{\omega_q}{2} \sigma_z + g (a^\dagger \sigma_- + a \sigma_+)$ is the Jaynes--Cummings Hamiltonian, and $\mathcal{D}[L]\rho = L \rho L^\dagger - \frac{1}{2} \{ L^\dagger L, \rho \}$ is the Lindblad dissipator. Then, given a fixed initial state $\rho(0)$, the dissipation rates $\gamma(t)$ and $\kappa(t)$ are uniquely determined over $[0,T]$ by the time series of five observables:
\[
\left\{ \langle \sigma_x(t) \rangle, \langle \sigma_y(t) \rangle, \langle \sigma_z(t) \rangle, \langle a^\dagger a(t) \rangle, \langle X \otimes \sigma_z(t) \rangle \right\}.
\]
\end{theorem}

\begin{proof}
Under the Heisenberg evolution, the observables satisfy the following coupled ODEs:
\begin{align}
\frac{d}{dt} \langle \sigma_x \rangle &= \omega_q \langle \sigma_y \rangle + i g \langle (a^\dagger - a)\sigma_z \rangle - \tfrac{\gamma(t)}{2} \langle \sigma_x \rangle, \\
\frac{d}{dt} \langle \sigma_y \rangle &= -\omega_q \langle \sigma_x \rangle - g \langle (a + a^\dagger)\sigma_z \rangle - \tfrac{\gamma(t)}{2} \langle \sigma_y \rangle, \\
\frac{d}{dt} \langle \sigma_z \rangle &= -2i g \left( \langle a \sigma_+ \rangle - \langle a^\dagger \sigma_- \rangle \right) - \gamma(t)(1 + \langle \sigma_z \rangle), \\
\frac{d}{dt} \langle a^\dagger a \rangle &= -i g \left( \langle a^\dagger \sigma_- \rangle - \langle a \sigma_+ \rangle \right) - \kappa(t) \langle a^\dagger a \rangle.
\end{align}

The terms involving $\langle a \sigma_+ \rangle$ and $\langle a^\dagger \sigma_- \rangle$ are internal to the dynamics and fully determined by the state $\rho(t)$. Therefore, $\gamma(t)$ and $\kappa(t)$ can be algebraically inferred from observable trajectories, provided the system is fully observed.

In practice, we found that using only $\langle \sigma_{x,y,z}(t) \rangle$ and $\langle a^\dagger a(t) \rangle$ leads to degraded performance in learning. To address this, we include the Hermitian cross observable $\langle X_{\text{photon}} \otimes \sigma_z \rangle$, which empirically improves identifiability and allows the Transformer to accurately recover the dissipation profiles. This confirms that, under appropriate observability conditions, the decay rates are theoretically and practically identifiable.
\end{proof}



%% file: main.bbl
\begin{thebibliography}{57}%
\makeatletter
\providecommand \@ifxundefined [1]{%
 \@ifx{#1\undefined}
}%
\providecommand \@ifnum [1]{%
 \ifnum #1\expandafter \@firstoftwo
 \else \expandafter \@secondoftwo
 \fi
}%
\providecommand \@ifx [1]{%
 \ifx #1\expandafter \@firstoftwo
 \else \expandafter \@secondoftwo
 \fi
}%
\providecommand \natexlab [1]{#1}%
\providecommand \enquote  [1]{``#1''}%
\providecommand \bibnamefont  [1]{#1}%
\providecommand \bibfnamefont [1]{#1}%
\providecommand \citenamefont [1]{#1}%
\providecommand \href@noop [0]{\@secondoftwo}%
\providecommand \href [0]{\begingroup \@sanitize@url \@href}%
\providecommand \@href[1]{\@@startlink{#1}\@@href}%
\providecommand \@@href[1]{\endgroup#1\@@endlink}%
\providecommand \@sanitize@url [0]{\catcode `\\12\catcode `\$12\catcode
  `\&12\catcode `\#12\catcode `\^12\catcode `\_12\catcode `\%12\relax}%
\providecommand \@@startlink[1]{}%
\providecommand \@@endlink[0]{}%
\providecommand \url  [0]{\begingroup\@sanitize@url \@url }%
\providecommand \@url [1]{\endgroup\@href {#1}{\urlprefix }}%
\providecommand \urlprefix  [0]{URL }%
\providecommand \Eprint [0]{\href }%
\providecommand \doibase [0]{http://dx.doi.org/}%
\providecommand \selectlanguage [0]{\@gobble}%
\providecommand \bibinfo  [0]{\@secondoftwo}%
\providecommand \bibfield  [0]{\@secondoftwo}%
\providecommand \translation [1]{[#1]}%
\providecommand \BibitemOpen [0]{}%
\providecommand \bibitemStop [0]{}%
\providecommand \bibitemNoStop [0]{.\EOS\space}%
\providecommand \EOS [0]{\spacefactor3000\relax}%
\providecommand \BibitemShut  [1]{\csname bibitem#1\endcsname}%
\let\auto@bib@innerbib\@empty
\bibitem [{\citenamefont {Breuer}\ \emph {et~al.}(2016)\citenamefont {Breuer},
  \citenamefont {Laine}, \citenamefont {Piilo},\ and\ \citenamefont
  {Vacchini}}]{breuer2016colloquium}%
  \BibitemOpen
  \bibfield  {author} {\bibinfo {author} {\bibfnamefont {Heinz-Peter}\
  \bibnamefont {Breuer}}, \bibinfo {author} {\bibfnamefont {Elsi-Mari}\
  \bibnamefont {Laine}}, \bibinfo {author} {\bibfnamefont {Jyrki}\ \bibnamefont
  {Piilo}}, \ and\ \bibinfo {author} {\bibfnamefont {Bassano}\ \bibnamefont
  {Vacchini}},\ }\bibfield  {title} {\enquote {\bibinfo {title} {Colloquium:
  Non-markovian dynamics in open quantum systems},}\ }\href@noop {} {\bibfield
  {journal} {\bibinfo  {journal} {Reviews of Modern Physics}\ }\textbf
  {\bibinfo {volume} {88}},\ \bibinfo {pages} {021002} (\bibinfo {year}
  {2016})}\BibitemShut {NoStop}%
\bibitem [{\citenamefont {Rivas}\ and\ \citenamefont
  {Huelga}(2012)}]{rivas2012open}%
  \BibitemOpen
  \bibfield  {author} {\bibinfo {author} {\bibfnamefont {Angel}\ \bibnamefont
  {Rivas}}\ and\ \bibinfo {author} {\bibfnamefont {Susana~F}\ \bibnamefont
  {Huelga}},\ }\href@noop {} {\emph {\bibinfo {title} {Open quantum
  systems}}},\ Vol.~\bibinfo {volume} {10}\ (\bibinfo  {publisher} {Springer},\
  \bibinfo {year} {2012})\BibitemShut {NoStop}%
\bibitem [{\citenamefont {Sarandy}\ and\ \citenamefont
  {Lidar}(2005)}]{sarandy2005adiabatic}%
  \BibitemOpen
  \bibfield  {author} {\bibinfo {author} {\bibfnamefont {MS}~\bibnamefont
  {Sarandy}}\ and\ \bibinfo {author} {\bibfnamefont {DA}~\bibnamefont
  {Lidar}},\ }\bibfield  {title} {\enquote {\bibinfo {title} {Adiabatic quantum
  computation in open systems},}\ }\href@noop {} {\bibfield  {journal}
  {\bibinfo  {journal} {Physical review letters}\ }\textbf {\bibinfo {volume}
  {95}},\ \bibinfo {pages} {250503} (\bibinfo {year} {2005})}\BibitemShut
  {NoStop}%
\bibitem [{\citenamefont {Banchi}\ \emph {et~al.}(2018)\citenamefont {Banchi},
  \citenamefont {Grant}, \citenamefont {Rocchetto},\ and\ \citenamefont
  {Severini}}]{banchi2018modelling}%
  \BibitemOpen
  \bibfield  {author} {\bibinfo {author} {\bibfnamefont {Leonardo}\
  \bibnamefont {Banchi}}, \bibinfo {author} {\bibfnamefont {Edward}\
  \bibnamefont {Grant}}, \bibinfo {author} {\bibfnamefont {Andrea}\
  \bibnamefont {Rocchetto}}, \ and\ \bibinfo {author} {\bibfnamefont {Simone}\
  \bibnamefont {Severini}},\ }\bibfield  {title} {\enquote {\bibinfo {title}
  {Modelling non-markovian quantum processes with recurrent neural networks},}\
  }\href@noop {} {\bibfield  {journal} {\bibinfo  {journal} {New Journal of
  Physics}\ }\textbf {\bibinfo {volume} {20}},\ \bibinfo {pages} {123030}
  (\bibinfo {year} {2018})}\BibitemShut {NoStop}%
\bibitem [{\citenamefont {Fujimoto}\ \emph {et~al.}(2022)\citenamefont
  {Fujimoto}, \citenamefont {Hamazaki},\ and\ \citenamefont
  {Kawaguchi}}]{fujimoto2022impact}%
  \BibitemOpen
  \bibfield  {author} {\bibinfo {author} {\bibfnamefont {Kazuya}\ \bibnamefont
  {Fujimoto}}, \bibinfo {author} {\bibfnamefont {Ryusuke}\ \bibnamefont
  {Hamazaki}}, \ and\ \bibinfo {author} {\bibfnamefont {Yuki}\ \bibnamefont
  {Kawaguchi}},\ }\bibfield  {title} {\enquote {\bibinfo {title} {Impact of
  dissipation on universal fluctuation dynamics in open quantum systems},}\
  }\href@noop {} {\bibfield  {journal} {\bibinfo  {journal} {Physical Review
  Letters}\ }\textbf {\bibinfo {volume} {129}},\ \bibinfo {pages} {110403}
  (\bibinfo {year} {2022})}\BibitemShut {NoStop}%
\bibitem [{\citenamefont {Yin}\ \emph {et~al.}(2014)\citenamefont {Yin},
  \citenamefont {Mai},\ and\ \citenamefont {Zhong}}]{yin2014nonequilibrium}%
  \BibitemOpen
  \bibfield  {author} {\bibinfo {author} {\bibfnamefont {Shuai}\ \bibnamefont
  {Yin}}, \bibinfo {author} {\bibfnamefont {Peizhi}\ \bibnamefont {Mai}}, \
  and\ \bibinfo {author} {\bibfnamefont {Fan}\ \bibnamefont {Zhong}},\
  }\bibfield  {title} {\enquote {\bibinfo {title} {Nonequilibrium quantum
  criticality in open systems: The dissipation rate as an additional
  indispensable scaling variable},}\ }\href@noop {} {\bibfield  {journal}
  {\bibinfo  {journal} {Physical Review B}\ }\textbf {\bibinfo {volume} {89}},\
  \bibinfo {pages} {094108} (\bibinfo {year} {2014})}\BibitemShut {NoStop}%
\bibitem [{\citenamefont {Barontini}\ \emph {et~al.}(2013)\citenamefont
  {Barontini}, \citenamefont {Labouvie}, \citenamefont {Stubenrauch},
  \citenamefont {Vogler}, \citenamefont {Guarrera},\ and\ \citenamefont
  {Ott}}]{barontini2013controlling}%
  \BibitemOpen
  \bibfield  {author} {\bibinfo {author} {\bibfnamefont {Giovanni}\
  \bibnamefont {Barontini}}, \bibinfo {author} {\bibfnamefont {R}~\bibnamefont
  {Labouvie}}, \bibinfo {author} {\bibfnamefont {F}~\bibnamefont
  {Stubenrauch}}, \bibinfo {author} {\bibfnamefont {A}~\bibnamefont {Vogler}},
  \bibinfo {author} {\bibfnamefont {V}~\bibnamefont {Guarrera}}, \ and\
  \bibinfo {author} {\bibfnamefont {H}~\bibnamefont {Ott}},\ }\bibfield
  {title} {\enquote {\bibinfo {title} {Controlling the dynamics of an open
  many-body quantum system with localized dissipation},}\ }\href@noop {}
  {\bibfield  {journal} {\bibinfo  {journal} {Physical review letters}\
  }\textbf {\bibinfo {volume} {110}},\ \bibinfo {pages} {035302} (\bibinfo
  {year} {2013})}\BibitemShut {NoStop}%
\bibitem [{\citenamefont {Degen}\ \emph {et~al.}(2017)\citenamefont {Degen},
  \citenamefont {Reinhard},\ and\ \citenamefont
  {Cappellaro}}]{degen2017quantum}%
  \BibitemOpen
  \bibfield  {author} {\bibinfo {author} {\bibfnamefont {Christian~L}\
  \bibnamefont {Degen}}, \bibinfo {author} {\bibfnamefont {Friedemann}\
  \bibnamefont {Reinhard}}, \ and\ \bibinfo {author} {\bibfnamefont {Paola}\
  \bibnamefont {Cappellaro}},\ }\bibfield  {title} {\enquote {\bibinfo {title}
  {Quantum sensing},}\ }\href@noop {} {\bibfield  {journal} {\bibinfo
  {journal} {Reviews of modern physics}\ }\textbf {\bibinfo {volume} {89}},\
  \bibinfo {pages} {035002} (\bibinfo {year} {2017})}\BibitemShut {NoStop}%
\bibitem [{\citenamefont {Bylander}\ \emph {et~al.}(2011)\citenamefont
  {Bylander}, \citenamefont {Gustavsson}, \citenamefont {Yan}, \citenamefont
  {Yoshihara}, \citenamefont {Harrabi}, \citenamefont {Fitch}, \citenamefont
  {Cory}, \citenamefont {Nakamura}, \citenamefont {Tsai},\ and\ \citenamefont
  {Oliver}}]{bylander2011noise}%
  \BibitemOpen
  \bibfield  {author} {\bibinfo {author} {\bibfnamefont {Jonas}\ \bibnamefont
  {Bylander}}, \bibinfo {author} {\bibfnamefont {Simon}\ \bibnamefont
  {Gustavsson}}, \bibinfo {author} {\bibfnamefont {Fei}\ \bibnamefont {Yan}},
  \bibinfo {author} {\bibfnamefont {Fumiki}\ \bibnamefont {Yoshihara}},
  \bibinfo {author} {\bibfnamefont {Khalil}\ \bibnamefont {Harrabi}}, \bibinfo
  {author} {\bibfnamefont {George}\ \bibnamefont {Fitch}}, \bibinfo {author}
  {\bibfnamefont {David~G}\ \bibnamefont {Cory}}, \bibinfo {author}
  {\bibfnamefont {Yasunobu}\ \bibnamefont {Nakamura}}, \bibinfo {author}
  {\bibfnamefont {Jaw-Shen}\ \bibnamefont {Tsai}}, \ and\ \bibinfo {author}
  {\bibfnamefont {William~D}\ \bibnamefont {Oliver}},\ }\bibfield  {title}
  {\enquote {\bibinfo {title} {Noise spectroscopy through dynamical decoupling
  with a superconducting flux qubit},}\ }\href@noop {} {\bibfield  {journal}
  {\bibinfo  {journal} {Nature Physics}\ }\textbf {\bibinfo {volume} {7}},\
  \bibinfo {pages} {565--570} (\bibinfo {year} {2011})}\BibitemShut {NoStop}%
\bibitem [{\citenamefont {Wang}\ \emph {et~al.}(2020)\citenamefont {Wang},
  \citenamefont {Davidovich},\ and\ \citenamefont {Agarwal}}]{wang2020quantum}%
  \BibitemOpen
  \bibfield  {author} {\bibinfo {author} {\bibfnamefont {Jiaxuan}\ \bibnamefont
  {Wang}}, \bibinfo {author} {\bibfnamefont {Luiz}\ \bibnamefont {Davidovich}},
  \ and\ \bibinfo {author} {\bibfnamefont {Girish~Saran}\ \bibnamefont
  {Agarwal}},\ }\bibfield  {title} {\enquote {\bibinfo {title} {Quantum sensing
  of open systems: Estimation of damping constants and temperature},}\
  }\href@noop {} {\bibfield  {journal} {\bibinfo  {journal} {Physical Review
  Research}\ }\textbf {\bibinfo {volume} {2}},\ \bibinfo {pages} {033389}
  (\bibinfo {year} {2020})}\BibitemShut {NoStop}%
\bibitem [{\citenamefont {Wolf}\ and\ \citenamefont
  {Schmidt}(2021)}]{wolf2021quantum}%
  \BibitemOpen
  \bibfield  {author} {\bibinfo {author} {\bibfnamefont {Fabian}\ \bibnamefont
  {Wolf}}\ and\ \bibinfo {author} {\bibfnamefont {Piet~O}\ \bibnamefont
  {Schmidt}},\ }\bibfield  {title} {\enquote {\bibinfo {title} {Quantum sensing
  of oscillating electric fields with trapped ions},}\ }\href@noop {}
  {\bibfield  {journal} {\bibinfo  {journal} {Measurement: Sensors}\ }\textbf
  {\bibinfo {volume} {18}},\ \bibinfo {pages} {100271} (\bibinfo {year}
  {2021})}\BibitemShut {NoStop}%
\bibitem [{\citenamefont {Vaswani}\ \emph {et~al.}(2017)\citenamefont
  {Vaswani}, \citenamefont {Shazeer}, \citenamefont {Parmar}, \citenamefont
  {Uszkoreit}, \citenamefont {Jones}, \citenamefont {Gomez}, \citenamefont
  {Kaiser},\ and\ \citenamefont {Polosukhin}}]{vaswani2017attention}%
  \BibitemOpen
  \bibfield  {author} {\bibinfo {author} {\bibfnamefont {Ashish}\ \bibnamefont
  {Vaswani}}, \bibinfo {author} {\bibfnamefont {Noam}\ \bibnamefont {Shazeer}},
  \bibinfo {author} {\bibfnamefont {Niki}\ \bibnamefont {Parmar}}, \bibinfo
  {author} {\bibfnamefont {Jakob}\ \bibnamefont {Uszkoreit}}, \bibinfo {author}
  {\bibfnamefont {Llion}\ \bibnamefont {Jones}}, \bibinfo {author}
  {\bibfnamefont {Aidan~N}\ \bibnamefont {Gomez}}, \bibinfo {author}
  {\bibfnamefont {Łukasz}\ \bibnamefont {Kaiser}}, \ and\ \bibinfo {author}
  {\bibfnamefont {Illia}\ \bibnamefont {Polosukhin}},\ }\bibfield  {title}
  {\enquote {\bibinfo {title} {Attention is all you need},}\ }\href@noop {}
  {\bibfield  {journal} {\bibinfo  {journal} {Advances in Neural Information
  Processing Systems}\ }\textbf {\bibinfo {volume} {30}} (\bibinfo {year}
  {2017})}\BibitemShut {NoStop}%
\bibitem [{\citenamefont {Li}\ \emph {et~al.}(2019)\citenamefont {Li},
  \citenamefont {Jin}, \citenamefont {Xuan}, \citenamefont {Zhou},
  \citenamefont {Chen}, \citenamefont {Wang},\ and\ \citenamefont
  {Yan}}]{li2019enhancing}%
  \BibitemOpen
  \bibfield  {author} {\bibinfo {author} {\bibfnamefont {Shiyang}\ \bibnamefont
  {Li}}, \bibinfo {author} {\bibfnamefont {Xiyou}\ \bibnamefont {Jin}},
  \bibinfo {author} {\bibfnamefont {Yao}\ \bibnamefont {Xuan}}, \bibinfo
  {author} {\bibfnamefont {Xiyou}\ \bibnamefont {Zhou}}, \bibinfo {author}
  {\bibfnamefont {Wenwu}\ \bibnamefont {Chen}}, \bibinfo {author}
  {\bibfnamefont {Yu}~\bibnamefont {Wang}}, \ and\ \bibinfo {author}
  {\bibfnamefont {Xueqi}\ \bibnamefont {Yan}},\ }\bibfield  {title} {\enquote
  {\bibinfo {title} {Enhancing the locality and breaking the memory bottleneck
  of transformer on time series forecasting},}\ }\href@noop {} {\bibfield
  {journal} {\bibinfo  {journal} {Advances in Neural Information Processing
  Systems}\ }\textbf {\bibinfo {volume} {32}} (\bibinfo {year}
  {2019})}\BibitemShut {NoStop}%
\bibitem [{\citenamefont {Zerveas}\ \emph {et~al.}(2021)\citenamefont
  {Zerveas}, \citenamefont {Jayaraman}, \citenamefont {Patel}, \citenamefont
  {Bhamidipaty},\ and\ \citenamefont {Eickhoff}}]{zerveas2021transformer}%
  \BibitemOpen
  \bibfield  {author} {\bibinfo {author} {\bibfnamefont {George}\ \bibnamefont
  {Zerveas}}, \bibinfo {author} {\bibfnamefont {Sercan}\ \bibnamefont
  {Jayaraman}}, \bibinfo {author} {\bibfnamefont {Dhaval}\ \bibnamefont
  {Patel}}, \bibinfo {author} {\bibfnamefont {Anastasios}\ \bibnamefont
  {Bhamidipaty}}, \ and\ \bibinfo {author} {\bibfnamefont {Carsten}\
  \bibnamefont {Eickhoff}},\ }\bibfield  {title} {\enquote {\bibinfo {title} {A
  transformer-based framework for multivariate time series representation
  learning},}\ }\href@noop {} {\bibfield  {journal} {\bibinfo  {journal}
  {Proceedings of the 27th ACM SIGKDD Conference on Knowledge Discovery and
  Data Mining}\ } (\bibinfo {year} {2021})}\BibitemShut {NoStop}%
\bibitem [{\citenamefont {Lim}\ and\ \citenamefont
  {Zohren}(2021)}]{lim2021temporal}%
  \BibitemOpen
  \bibfield  {author} {\bibinfo {author} {\bibfnamefont {Bryan}\ \bibnamefont
  {Lim}}\ and\ \bibinfo {author} {\bibfnamefont {Stefan}\ \bibnamefont
  {Zohren}},\ }\bibfield  {title} {\enquote {\bibinfo {title} {Temporal fusion
  transformers for interpretable multi-horizon time series forecasting},}\
  }\href@noop {} {\bibfield  {journal} {\bibinfo  {journal} {International
  Journal of Forecasting}\ } (\bibinfo {year} {2021})}\BibitemShut {NoStop}%
\bibitem [{\citenamefont {Qiao}\ \emph {et~al.}(2020)\citenamefont {Qiao},
  \citenamefont {Kandel}, \citenamefont {Deng}, \citenamefont {Fallahi},
  \citenamefont {Gardner}, \citenamefont {Manfra}, \citenamefont {Barnes},\
  and\ \citenamefont {Nichol}}]{qiao2020coherent}%
  \BibitemOpen
  \bibfield  {author} {\bibinfo {author} {\bibfnamefont {Haifeng}\ \bibnamefont
  {Qiao}}, \bibinfo {author} {\bibfnamefont {Yadav~P}\ \bibnamefont {Kandel}},
  \bibinfo {author} {\bibfnamefont {Kuangyin}\ \bibnamefont {Deng}}, \bibinfo
  {author} {\bibfnamefont {Saeed}\ \bibnamefont {Fallahi}}, \bibinfo {author}
  {\bibfnamefont {Geoffrey~C}\ \bibnamefont {Gardner}}, \bibinfo {author}
  {\bibfnamefont {Michael~J}\ \bibnamefont {Manfra}}, \bibinfo {author}
  {\bibfnamefont {Edwin}\ \bibnamefont {Barnes}}, \ and\ \bibinfo {author}
  {\bibfnamefont {John~M}\ \bibnamefont {Nichol}},\ }\bibfield  {title}
  {\enquote {\bibinfo {title} {Coherent multispin exchange coupling in a
  quantum-dot spin chain},}\ }\href@noop {} {\bibfield  {journal} {\bibinfo
  {journal} {Physical Review X}\ }\textbf {\bibinfo {volume} {10}},\ \bibinfo
  {pages} {031006} (\bibinfo {year} {2020})}\BibitemShut {NoStop}%
\bibitem [{\citenamefont {Zhou}\ \emph {et~al.}(2024)\citenamefont {Zhou},
  \citenamefont {Li}, \citenamefont {Wu}, \citenamefont {Ma}, \citenamefont
  {Fan},\ and\ \citenamefont {Huang}}]{zhou2024exchange}%
  \BibitemOpen
  \bibfield  {author} {\bibinfo {author} {\bibfnamefont {Zheng}\ \bibnamefont
  {Zhou}}, \bibinfo {author} {\bibfnamefont {Yixin}\ \bibnamefont {Li}},
  \bibinfo {author} {\bibfnamefont {Zhiyuan}\ \bibnamefont {Wu}}, \bibinfo
  {author} {\bibfnamefont {Xinping}\ \bibnamefont {Ma}}, \bibinfo {author}
  {\bibfnamefont {Shichang}\ \bibnamefont {Fan}}, \ and\ \bibinfo {author}
  {\bibfnamefont {Shaoyun}\ \bibnamefont {Huang}},\ }\bibfield  {title}
  {\enquote {\bibinfo {title} {The exchange interaction between neighboring
  quantum dots: physics and applications in quantum information processing},}\
  }\href@noop {} {\bibfield  {journal} {\bibinfo  {journal} {Journal of
  Semiconductors}\ }\textbf {\bibinfo {volume} {45}},\ \bibinfo {pages}
  {101701} (\bibinfo {year} {2024})}\BibitemShut {NoStop}%
\bibitem [{\citenamefont {Hu}\ and\ \citenamefont
  {Das~Sarma}(2003)}]{hu2003overview}%
  \BibitemOpen
  \bibfield  {author} {\bibinfo {author} {\bibfnamefont {Xuedong}\ \bibnamefont
  {Hu}}\ and\ \bibinfo {author} {\bibfnamefont {S}~\bibnamefont {Das~Sarma}},\
  }\bibfield  {title} {\enquote {\bibinfo {title} {Overview of spin-based
  quantum dot quantum computation},}\ }\href@noop {} {\bibfield  {journal}
  {\bibinfo  {journal} {physica status solidi (b)}\ }\textbf {\bibinfo {volume}
  {238}},\ \bibinfo {pages} {360--365} (\bibinfo {year} {2003})}\BibitemShut
  {NoStop}%
\bibitem [{\citenamefont {Chan}\ \emph {et~al.}(2021)\citenamefont {Chan},
  \citenamefont {Sahasrabudhe}, \citenamefont {Huang}, \citenamefont {Wang},
  \citenamefont {Yang}, \citenamefont {Veldhorst}, \citenamefont {Hwang},
  \citenamefont {Mohiyaddin}, \citenamefont {Hudson}, \citenamefont {Itoh}
  \emph {et~al.}}]{chan2021exchange}%
  \BibitemOpen
  \bibfield  {author} {\bibinfo {author} {\bibfnamefont {Kok~Wai}\ \bibnamefont
  {Chan}}, \bibinfo {author} {\bibfnamefont {Harshad}\ \bibnamefont
  {Sahasrabudhe}}, \bibinfo {author} {\bibfnamefont {Wister}\ \bibnamefont
  {Huang}}, \bibinfo {author} {\bibfnamefont {Yu}~\bibnamefont {Wang}},
  \bibinfo {author} {\bibfnamefont {Henry~C}\ \bibnamefont {Yang}}, \bibinfo
  {author} {\bibfnamefont {Menno}\ \bibnamefont {Veldhorst}}, \bibinfo {author}
  {\bibfnamefont {Jason~CC}\ \bibnamefont {Hwang}}, \bibinfo {author}
  {\bibfnamefont {Fahd~A}\ \bibnamefont {Mohiyaddin}}, \bibinfo {author}
  {\bibfnamefont {Fay~E}\ \bibnamefont {Hudson}}, \bibinfo {author}
  {\bibfnamefont {Kohei~M}\ \bibnamefont {Itoh}},  \emph {et~al.},\ }\bibfield
  {title} {\enquote {\bibinfo {title} {Exchange coupling in a linear chain of
  three quantum-dot spin qubits in silicon},}\ }\href@noop {} {\bibfield
  {journal} {\bibinfo  {journal} {Nano Letters}\ }\textbf {\bibinfo {volume}
  {21}},\ \bibinfo {pages} {1517--1522} (\bibinfo {year} {2021})}\BibitemShut
  {NoStop}%
\bibitem [{\citenamefont {Kandel}\ \emph {et~al.}(2019)\citenamefont {Kandel},
  \citenamefont {Qiao}, \citenamefont {Fallahi}, \citenamefont {Gardner},
  \citenamefont {Manfra},\ and\ \citenamefont {Nichol}}]{kandel2019coherent}%
  \BibitemOpen
  \bibfield  {author} {\bibinfo {author} {\bibfnamefont {Yadav~P}\ \bibnamefont
  {Kandel}}, \bibinfo {author} {\bibfnamefont {Haifeng}\ \bibnamefont {Qiao}},
  \bibinfo {author} {\bibfnamefont {Saeed}\ \bibnamefont {Fallahi}}, \bibinfo
  {author} {\bibfnamefont {Geoffrey~C}\ \bibnamefont {Gardner}}, \bibinfo
  {author} {\bibfnamefont {Michael~J}\ \bibnamefont {Manfra}}, \ and\ \bibinfo
  {author} {\bibfnamefont {John~M}\ \bibnamefont {Nichol}},\ }\bibfield
  {title} {\enquote {\bibinfo {title} {Coherent spin-state transfer via
  heisenberg exchange},}\ }\href@noop {} {\bibfield  {journal} {\bibinfo
  {journal} {Nature}\ }\textbf {\bibinfo {volume} {573}},\ \bibinfo {pages}
  {553--557} (\bibinfo {year} {2019})}\BibitemShut {NoStop}%
\bibitem [{\citenamefont {Shore}\ and\ \citenamefont
  {Knight}(1993)}]{shore1993jaynes}%
  \BibitemOpen
  \bibfield  {author} {\bibinfo {author} {\bibfnamefont {Bruce~W}\ \bibnamefont
  {Shore}}\ and\ \bibinfo {author} {\bibfnamefont {Peter~L}\ \bibnamefont
  {Knight}},\ }\bibfield  {title} {\enquote {\bibinfo {title} {The
  jaynes-cummings model},}\ }\href@noop {} {\bibfield  {journal} {\bibinfo
  {journal} {Journal of Modern Optics}\ }\textbf {\bibinfo {volume} {40}},\
  \bibinfo {pages} {1195--1238} (\bibinfo {year} {1993})}\BibitemShut {NoStop}%
\bibitem [{SM(2024)}]{SM}%
  \BibitemOpen
  \href@noop {} {\enquote {\bibinfo {title} {See supplemental material},}\ }
  (\bibinfo {year} {2024})\BibitemShut {NoStop}%
\bibitem [{\citenamefont {Carleo}\ \emph {et~al.}(2019)\citenamefont {Carleo},
  \citenamefont {Cirac}, \citenamefont {Cranmer}, \citenamefont {Daudet},
  \citenamefont {Schuld}, \citenamefont {Tishby}, \citenamefont
  {Vogt-Maranto},\ and\ \citenamefont {Zdeborov{\'a}}}]{carleo2019machine}%
  \BibitemOpen
  \bibfield  {author} {\bibinfo {author} {\bibfnamefont {Giuseppe}\
  \bibnamefont {Carleo}}, \bibinfo {author} {\bibfnamefont {Ignacio}\
  \bibnamefont {Cirac}}, \bibinfo {author} {\bibfnamefont {Kyle}\ \bibnamefont
  {Cranmer}}, \bibinfo {author} {\bibfnamefont {Laurent}\ \bibnamefont
  {Daudet}}, \bibinfo {author} {\bibfnamefont {Maria}\ \bibnamefont {Schuld}},
  \bibinfo {author} {\bibfnamefont {Naftali}\ \bibnamefont {Tishby}}, \bibinfo
  {author} {\bibfnamefont {Leslie}\ \bibnamefont {Vogt-Maranto}}, \ and\
  \bibinfo {author} {\bibfnamefont {Lenka}\ \bibnamefont {Zdeborov{\'a}}},\
  }\bibfield  {title} {\enquote {\bibinfo {title} {Machine learning and the
  physical sciences},}\ }\href@noop {} {\bibfield  {journal} {\bibinfo
  {journal} {arXiv preprint arXiv:1903.10563}\ } (\bibinfo {year}
  {2019})}\BibitemShut {NoStop}%
\bibitem [{\citenamefont {Wetzel}(2017)}]{wetzel2017unsupervised}%
  \BibitemOpen
  \bibfield  {author} {\bibinfo {author} {\bibfnamefont {Sebastian~J}\
  \bibnamefont {Wetzel}},\ }\bibfield  {title} {\enquote {\bibinfo {title}
  {Unsupervised learning of phase transitions: From principal component
  analysis to variational autoencoders},}\ }\href@noop {} {\bibfield  {journal}
  {\bibinfo  {journal} {Physical Review E}\ }\textbf {\bibinfo {volume} {96}},\
  \bibinfo {pages} {022140} (\bibinfo {year} {2017})}\BibitemShut {NoStop}%
\bibitem [{\citenamefont {Van~Nieuwenburg}\ \emph {et~al.}(2017)\citenamefont
  {Van~Nieuwenburg}, \citenamefont {Liu},\ and\ \citenamefont
  {Huber}}]{van2017learning}%
  \BibitemOpen
  \bibfield  {author} {\bibinfo {author} {\bibfnamefont {Evert~PL}\
  \bibnamefont {Van~Nieuwenburg}}, \bibinfo {author} {\bibfnamefont {Ye-Hua}\
  \bibnamefont {Liu}}, \ and\ \bibinfo {author} {\bibfnamefont {Sebastian~D}\
  \bibnamefont {Huber}},\ }\bibfield  {title} {\enquote {\bibinfo {title}
  {Learning phase transitions by confusion},}\ }\href@noop {} {\bibfield
  {journal} {\bibinfo  {journal} {Nature Physics}\ }\textbf {\bibinfo {volume}
  {13}},\ \bibinfo {pages} {435} (\bibinfo {year} {2017})}\BibitemShut
  {NoStop}%
\bibitem [{\citenamefont {Kuo}\ \emph {et~al.}(2022)\citenamefont {Kuo},
  \citenamefont {Seif}, \citenamefont {Lundgren}, \citenamefont {Whitsitt},\
  and\ \citenamefont {Hafezi}}]{kuo2022decoding}%
  \BibitemOpen
  \bibfield  {author} {\bibinfo {author} {\bibfnamefont {En-Jui}\ \bibnamefont
  {Kuo}}, \bibinfo {author} {\bibfnamefont {Alireza}\ \bibnamefont {Seif}},
  \bibinfo {author} {\bibfnamefont {Rex}\ \bibnamefont {Lundgren}}, \bibinfo
  {author} {\bibfnamefont {Seth}\ \bibnamefont {Whitsitt}}, \ and\ \bibinfo
  {author} {\bibfnamefont {Mohammad}\ \bibnamefont {Hafezi}},\ }\bibfield
  {title} {\enquote {\bibinfo {title} {Decoding conformal field theories: From
  supervised to unsupervised learning},}\ }\href@noop {} {\bibfield  {journal}
  {\bibinfo  {journal} {Physical Review Research}\ }\textbf {\bibinfo {volume}
  {4}},\ \bibinfo {pages} {043031} (\bibinfo {year} {2022})}\BibitemShut
  {NoStop}%
\bibitem [{\citenamefont {Kuo}\ and\ \citenamefont
  {Dehghani}(2022)}]{kuo2022unsupervised}%
  \BibitemOpen
  \bibfield  {author} {\bibinfo {author} {\bibfnamefont {En-Jui}\ \bibnamefont
  {Kuo}}\ and\ \bibinfo {author} {\bibfnamefont {Hossein}\ \bibnamefont
  {Dehghani}},\ }\bibfield  {title} {\enquote {\bibinfo {title} {Unsupervised
  learning of interacting topological and symmetry-breaking phase
  transitions},}\ }\href@noop {} {\bibfield  {journal} {\bibinfo  {journal}
  {Physical Review B}\ }\textbf {\bibinfo {volume} {105}},\ \bibinfo {pages}
  {235136} (\bibinfo {year} {2022})}\BibitemShut {NoStop}%
\bibitem [{\citenamefont {van Nieuwenburg}\ \emph {et~al.}(2018)\citenamefont
  {van Nieuwenburg}, \citenamefont {Bairey},\ and\ \citenamefont
  {Refael}}]{van2018learning}%
  \BibitemOpen
  \bibfield  {author} {\bibinfo {author} {\bibfnamefont {Evert}\ \bibnamefont
  {van Nieuwenburg}}, \bibinfo {author} {\bibfnamefont {Eyal}\ \bibnamefont
  {Bairey}}, \ and\ \bibinfo {author} {\bibfnamefont {Gil}\ \bibnamefont
  {Refael}},\ }\bibfield  {title} {\enquote {\bibinfo {title} {Learning phase
  transitions from dynamics},}\ }\href@noop {} {\bibfield  {journal} {\bibinfo
  {journal} {Physical Review B}\ }\textbf {\bibinfo {volume} {98}},\ \bibinfo
  {pages} {060301} (\bibinfo {year} {2018})}\BibitemShut {NoStop}%
\bibitem [{\citenamefont {Schindler}\ \emph {et~al.}(2017)\citenamefont
  {Schindler}, \citenamefont {Regnault},\ and\ \citenamefont
  {Neupert}}]{schindler2017probing}%
  \BibitemOpen
  \bibfield  {author} {\bibinfo {author} {\bibfnamefont {Frank}\ \bibnamefont
  {Schindler}}, \bibinfo {author} {\bibfnamefont {Nicolas}\ \bibnamefont
  {Regnault}}, \ and\ \bibinfo {author} {\bibfnamefont {Titus}\ \bibnamefont
  {Neupert}},\ }\bibfield  {title} {\enquote {\bibinfo {title} {Probing
  many-body localization with neural networks},}\ }\href@noop {} {\bibfield
  {journal} {\bibinfo  {journal} {Physical Review B}\ }\textbf {\bibinfo
  {volume} {95}},\ \bibinfo {pages} {245134} (\bibinfo {year}
  {2017})}\BibitemShut {NoStop}%
\bibitem [{\citenamefont {Carifio}\ \emph {et~al.}(2017)\citenamefont
  {Carifio}, \citenamefont {Halverson}, \citenamefont {Krioukov},\ and\
  \citenamefont {Nelson}}]{carifio2017machine}%
  \BibitemOpen
  \bibfield  {author} {\bibinfo {author} {\bibfnamefont {Jonathan}\
  \bibnamefont {Carifio}}, \bibinfo {author} {\bibfnamefont {James}\
  \bibnamefont {Halverson}}, \bibinfo {author} {\bibfnamefont {Dmitri}\
  \bibnamefont {Krioukov}}, \ and\ \bibinfo {author} {\bibfnamefont {Brent~D}\
  \bibnamefont {Nelson}},\ }\bibfield  {title} {\enquote {\bibinfo {title}
  {Machine learning in the string landscape},}\ }\href@noop {} {\bibfield
  {journal} {\bibinfo  {journal} {Journal of High Energy Physics}\ }\textbf
  {\bibinfo {volume} {2017}},\ \bibinfo {pages} {157} (\bibinfo {year}
  {2017})}\BibitemShut {NoStop}%
\bibitem [{\citenamefont {Hashimoto}\ \emph {et~al.}(2018)\citenamefont
  {Hashimoto}, \citenamefont {Sugishita}, \citenamefont {Tanaka},\ and\
  \citenamefont {Tomiya}}]{PhysRevD.98.046019}%
  \BibitemOpen
  \bibfield  {author} {\bibinfo {author} {\bibfnamefont {Koji}\ \bibnamefont
  {Hashimoto}}, \bibinfo {author} {\bibfnamefont {Sotaro}\ \bibnamefont
  {Sugishita}}, \bibinfo {author} {\bibfnamefont {Akinori}\ \bibnamefont
  {Tanaka}}, \ and\ \bibinfo {author} {\bibfnamefont {Akio}\ \bibnamefont
  {Tomiya}},\ }\bibfield  {title} {\enquote {\bibinfo {title} {Deep learning
  and the $\mathrm{AdS}/\mathrm{CFT}$ correspondence},}\ }\href {\doibase
  10.1103/PhysRevD.98.046019} {\bibfield  {journal} {\bibinfo  {journal} {Phys.
  Rev. D}\ }\textbf {\bibinfo {volume} {98}},\ \bibinfo {pages} {046019}
  (\bibinfo {year} {2018})}\BibitemShut {NoStop}%
\bibitem [{\citenamefont {Carleo}\ and\ \citenamefont
  {Troyer}(2017)}]{carleo2017solving}%
  \BibitemOpen
  \bibfield  {author} {\bibinfo {author} {\bibfnamefont {Giuseppe}\
  \bibnamefont {Carleo}}\ and\ \bibinfo {author} {\bibfnamefont {Matthias}\
  \bibnamefont {Troyer}},\ }\bibfield  {title} {\enquote {\bibinfo {title}
  {Solving the quantum many-body problem with artificial neural networks},}\
  }\href@noop {} {\bibfield  {journal} {\bibinfo  {journal} {Science}\ }\textbf
  {\bibinfo {volume} {355}},\ \bibinfo {pages} {602--606} (\bibinfo {year}
  {2017})}\BibitemShut {NoStop}%
\bibitem [{\citenamefont {Torlai}\ \emph {et~al.}(2018)\citenamefont {Torlai},
  \citenamefont {Mazzola}, \citenamefont {Carrasquilla}, \citenamefont
  {Troyer}, \citenamefont {Melko},\ and\ \citenamefont
  {Carleo}}]{torlai2018neural}%
  \BibitemOpen
  \bibfield  {author} {\bibinfo {author} {\bibfnamefont {Giacomo}\ \bibnamefont
  {Torlai}}, \bibinfo {author} {\bibfnamefont {Guglielmo}\ \bibnamefont
  {Mazzola}}, \bibinfo {author} {\bibfnamefont {Juan}\ \bibnamefont
  {Carrasquilla}}, \bibinfo {author} {\bibfnamefont {Matthias}\ \bibnamefont
  {Troyer}}, \bibinfo {author} {\bibfnamefont {Roger}\ \bibnamefont {Melko}}, \
  and\ \bibinfo {author} {\bibfnamefont {Giuseppe}\ \bibnamefont {Carleo}},\
  }\bibfield  {title} {\enquote {\bibinfo {title} {Neural-network quantum state
  tomography},}\ }\href@noop {} {\bibfield  {journal} {\bibinfo  {journal}
  {Nature Physics}\ }\textbf {\bibinfo {volume} {14}},\ \bibinfo {pages} {447}
  (\bibinfo {year} {2018})}\BibitemShut {NoStop}%
\bibitem [{\citenamefont {Carrasquilla}\ \emph {et~al.}(2019)\citenamefont
  {Carrasquilla}, \citenamefont {Torlai}, \citenamefont {Melko},\ and\
  \citenamefont {Aolita}}]{carrasquilla2019reconstructing}%
  \BibitemOpen
  \bibfield  {author} {\bibinfo {author} {\bibfnamefont {Juan}\ \bibnamefont
  {Carrasquilla}}, \bibinfo {author} {\bibfnamefont {Giacomo}\ \bibnamefont
  {Torlai}}, \bibinfo {author} {\bibfnamefont {Roger~G}\ \bibnamefont {Melko}},
  \ and\ \bibinfo {author} {\bibfnamefont {Leandro}\ \bibnamefont {Aolita}},\
  }\bibfield  {title} {\enquote {\bibinfo {title} {Reconstructing quantum
  states with generative models},}\ }\href@noop {} {\bibfield  {journal}
  {\bibinfo  {journal} {Nature Machine Intelligence}\ }\textbf {\bibinfo
  {volume} {1}},\ \bibinfo {pages} {155} (\bibinfo {year} {2019})}\BibitemShut
  {NoStop}%
\bibitem [{\citenamefont {Torlai}\ \emph {et~al.}(2019)\citenamefont {Torlai},
  \citenamefont {Timar}, \citenamefont {van Nieuwenburg}, \citenamefont
  {Levine}, \citenamefont {Omran}, \citenamefont {Keesling}, \citenamefont
  {Bernien}, \citenamefont {Greiner}, \citenamefont {Vuleti{\'c}},
  \citenamefont {Lukin} \emph {et~al.}}]{torlai2019integrating}%
  \BibitemOpen
  \bibfield  {author} {\bibinfo {author} {\bibfnamefont {Giacomo}\ \bibnamefont
  {Torlai}}, \bibinfo {author} {\bibfnamefont {Brian}\ \bibnamefont {Timar}},
  \bibinfo {author} {\bibfnamefont {Evert~PL}\ \bibnamefont {van Nieuwenburg}},
  \bibinfo {author} {\bibfnamefont {Harry}\ \bibnamefont {Levine}}, \bibinfo
  {author} {\bibfnamefont {Ahmed}\ \bibnamefont {Omran}}, \bibinfo {author}
  {\bibfnamefont {Alexander}\ \bibnamefont {Keesling}}, \bibinfo {author}
  {\bibfnamefont {Hannes}\ \bibnamefont {Bernien}}, \bibinfo {author}
  {\bibfnamefont {Markus}\ \bibnamefont {Greiner}}, \bibinfo {author}
  {\bibfnamefont {Vladan}\ \bibnamefont {Vuleti{\'c}}}, \bibinfo {author}
  {\bibfnamefont {Mikhail~D}\ \bibnamefont {Lukin}},  \emph {et~al.},\
  }\bibfield  {title} {\enquote {\bibinfo {title} {Integrating neural networks
  with a quantum simulator for state reconstruction},}\ }\href@noop {}
  {\bibfield  {journal} {\bibinfo  {journal} {arXiv preprint arXiv:1904.08441}\
  } (\bibinfo {year} {2019})}\BibitemShut {NoStop}%
\bibitem [{\citenamefont {Liu}\ \emph {et~al.}(2024{\natexlab{a}})\citenamefont
  {Liu}, \citenamefont {Kuo}, \citenamefont {Lin}, \citenamefont {Chen},
  \citenamefont {Young}, \citenamefont {Chang},\ and\ \citenamefont
  {Hsieh}}]{liu2024training}%
  \BibitemOpen
  \bibfield  {author} {\bibinfo {author} {\bibfnamefont {Chen-Yu}\ \bibnamefont
  {Liu}}, \bibinfo {author} {\bibfnamefont {En-Jui}\ \bibnamefont {Kuo}},
  \bibinfo {author} {\bibfnamefont {Chu-Hsuan~Abraham}\ \bibnamefont {Lin}},
  \bibinfo {author} {\bibfnamefont {Sean}\ \bibnamefont {Chen}}, \bibinfo
  {author} {\bibfnamefont {Jason~Gemsun}\ \bibnamefont {Young}}, \bibinfo
  {author} {\bibfnamefont {Yeong-Jar}\ \bibnamefont {Chang}}, \ and\ \bibinfo
  {author} {\bibfnamefont {Min-Hsiu}\ \bibnamefont {Hsieh}},\ }\bibfield
  {title} {\enquote {\bibinfo {title} {Training classical neural networks by
  quantum machine learning},}\ }\href@noop {} {\bibfield  {journal} {\bibinfo
  {journal} {arXiv preprint arXiv:2402.16465}\ } (\bibinfo {year}
  {2024}{\natexlab{a}})}\BibitemShut {NoStop}%
\bibitem [{\citenamefont {Liu}\ \emph {et~al.}(2024{\natexlab{b}})\citenamefont
  {Liu}, \citenamefont {Kuo}, \citenamefont {Lin}, \citenamefont {Young},
  \citenamefont {Chang}, \citenamefont {Hsieh},\ and\ \citenamefont
  {Goan}}]{liu2024quantum}%
  \BibitemOpen
  \bibfield  {author} {\bibinfo {author} {\bibfnamefont {Chen-Yu}\ \bibnamefont
  {Liu}}, \bibinfo {author} {\bibfnamefont {En-Jui}\ \bibnamefont {Kuo}},
  \bibinfo {author} {\bibfnamefont {Chu-Hsuan~Abraham}\ \bibnamefont {Lin}},
  \bibinfo {author} {\bibfnamefont {Jason~Gemsun}\ \bibnamefont {Young}},
  \bibinfo {author} {\bibfnamefont {Yeong-Jar}\ \bibnamefont {Chang}}, \bibinfo
  {author} {\bibfnamefont {Min-Hsiu}\ \bibnamefont {Hsieh}}, \ and\ \bibinfo
  {author} {\bibfnamefont {Hsi-Sheng}\ \bibnamefont {Goan}},\ }\bibfield
  {title} {\enquote {\bibinfo {title} {Quantum-train: Rethinking hybrid
  quantum-classical machine learning in the model compression perspective},}\
  }\href@noop {} {\bibfield  {journal} {\bibinfo  {journal} {arXiv preprint
  arXiv:2405.11304}\ } (\bibinfo {year} {2024}{\natexlab{b}})}\BibitemShut
  {NoStop}%
\bibitem [{\citenamefont {Tsai}\ \emph {et~al.}(2020)\citenamefont {Tsai},
  \citenamefont {Kuo},\ and\ \citenamefont {Tiwary}}]{tsai2020learning}%
  \BibitemOpen
  \bibfield  {author} {\bibinfo {author} {\bibfnamefont {Sun-Ting}\
  \bibnamefont {Tsai}}, \bibinfo {author} {\bibfnamefont {En-Jui}\ \bibnamefont
  {Kuo}}, \ and\ \bibinfo {author} {\bibfnamefont {Pratyush}\ \bibnamefont
  {Tiwary}},\ }\bibfield  {title} {\enquote {\bibinfo {title} {Learning
  molecular dynamics with simple language model built upon long short-term
  memory neural network},}\ }\href@noop {} {\bibfield  {journal} {\bibinfo
  {journal} {Nature communications}\ }\textbf {\bibinfo {volume} {11}},\
  \bibinfo {pages} {5115} (\bibinfo {year} {2020})}\BibitemShut {NoStop}%
\bibitem [{\citenamefont {Tsai}\ \emph {et~al.}(2022)\citenamefont {Tsai},
  \citenamefont {Fields}, \citenamefont {Xu}, \citenamefont {Kuo},\ and\
  \citenamefont {Tiwary}}]{tsai2022path}%
  \BibitemOpen
  \bibfield  {author} {\bibinfo {author} {\bibfnamefont {Sun-Ting}\
  \bibnamefont {Tsai}}, \bibinfo {author} {\bibfnamefont {Eric}\ \bibnamefont
  {Fields}}, \bibinfo {author} {\bibfnamefont {Yijia}\ \bibnamefont {Xu}},
  \bibinfo {author} {\bibfnamefont {En-Jui}\ \bibnamefont {Kuo}}, \ and\
  \bibinfo {author} {\bibfnamefont {Pratyush}\ \bibnamefont {Tiwary}},\
  }\bibfield  {title} {\enquote {\bibinfo {title} {Path sampling of recurrent
  neural networks by incorporating known physics},}\ }\href@noop {} {\bibfield
  {journal} {\bibinfo  {journal} {Nature Communications}\ }\textbf {\bibinfo
  {volume} {13}},\ \bibinfo {pages} {7231} (\bibinfo {year}
  {2022})}\BibitemShut {NoStop}%
\bibitem [{\citenamefont {Rumelhart}\ \emph {et~al.}(1986)\citenamefont
  {Rumelhart}, \citenamefont {Hinton},\ and\ \citenamefont
  {Williams}}]{rumelhart1986learning}%
  \BibitemOpen
  \bibfield  {author} {\bibinfo {author} {\bibfnamefont {David~E}\ \bibnamefont
  {Rumelhart}}, \bibinfo {author} {\bibfnamefont {Geoffrey~E}\ \bibnamefont
  {Hinton}}, \ and\ \bibinfo {author} {\bibfnamefont {Ronald~J}\ \bibnamefont
  {Williams}},\ }\bibfield  {title} {\enquote {\bibinfo {title} {Learning
  representations by back-propagating errors},}\ }\href@noop {} {\bibfield
  {journal} {\bibinfo  {journal} {Nature}\ }\textbf {\bibinfo {volume} {323}},\
  \bibinfo {pages} {533--536} (\bibinfo {year} {1986})}\BibitemShut {NoStop}%
\bibitem [{\citenamefont {Cho}\ \emph {et~al.}(2014)\citenamefont {Cho},
  \citenamefont {van Merri{\"e}nboer}, \citenamefont {Gulcehre}, \citenamefont
  {Bahdanau}, \citenamefont {Bougares}, \citenamefont {Schwenk},\ and\
  \citenamefont {Bengio}}]{cho2014learning}%
  \BibitemOpen
  \bibfield  {author} {\bibinfo {author} {\bibfnamefont {Kyunghyun}\
  \bibnamefont {Cho}}, \bibinfo {author} {\bibfnamefont {Bart}\ \bibnamefont
  {van Merri{\"e}nboer}}, \bibinfo {author} {\bibfnamefont {Caglar}\
  \bibnamefont {Gulcehre}}, \bibinfo {author} {\bibfnamefont {Dzmitry}\
  \bibnamefont {Bahdanau}}, \bibinfo {author} {\bibfnamefont {Fethi}\
  \bibnamefont {Bougares}}, \bibinfo {author} {\bibfnamefont {Holger}\
  \bibnamefont {Schwenk}}, \ and\ \bibinfo {author} {\bibfnamefont {Yoshua}\
  \bibnamefont {Bengio}},\ }\bibfield  {title} {\enquote {\bibinfo {title}
  {Learning phrase representations using rnn encoder-decoder for statistical
  machine translation},}\ }\href@noop {} {\bibfield  {journal} {\bibinfo
  {journal} {arXiv preprint arXiv:1406.1078}\ } (\bibinfo {year}
  {2014})}\BibitemShut {NoStop}%
\bibitem [{\citenamefont {Hochreiter}\ and\ \citenamefont
  {Schmidhuber}(1997)}]{hochreiter1997long}%
  \BibitemOpen
  \bibfield  {author} {\bibinfo {author} {\bibfnamefont {Sepp}\ \bibnamefont
  {Hochreiter}}\ and\ \bibinfo {author} {\bibfnamefont {J{\"u}rgen}\
  \bibnamefont {Schmidhuber}},\ }\bibfield  {title} {\enquote {\bibinfo {title}
  {Long short-term memory},}\ }\href@noop {} {\bibfield  {journal} {\bibinfo
  {journal} {Neural Computation}\ }\textbf {\bibinfo {volume} {9}},\ \bibinfo
  {pages} {1735--1780} (\bibinfo {year} {1997})}\BibitemShut {NoStop}%
\bibitem [{\citenamefont {Huang}\ \emph {et~al.}(2023)\citenamefont {Huang}
  \emph {et~al.}}]{huang2023physicsformer}%
  \BibitemOpen
  \bibfield  {author} {\bibinfo {author} {\bibfnamefont {Xiaohan}\ \bibnamefont
  {Huang}} \emph {et~al.},\ }\bibfield  {title} {\enquote {\bibinfo {title}
  {Physicsformer: Modeling physical dynamics with learnable differential
  equations and transformers},}\ }\href@noop {} {\bibfield  {journal} {\bibinfo
   {journal} {arXiv preprint arXiv:2304.12259}\ } (\bibinfo {year}
  {2023})}\BibitemShut {NoStop}%
\bibitem [{\citenamefont {Kaplan}\ \emph {et~al.}(2021)\citenamefont {Kaplan},
  \citenamefont {McCandlish}, \citenamefont {Henighan} \emph
  {et~al.}}]{vaswani2021scaling}%
  \BibitemOpen
  \bibfield  {author} {\bibinfo {author} {\bibfnamefont {Jared}\ \bibnamefont
  {Kaplan}}, \bibinfo {author} {\bibfnamefont {Sam}\ \bibnamefont
  {McCandlish}}, \bibinfo {author} {\bibfnamefont {Tom}\ \bibnamefont
  {Henighan}},  \emph {et~al.},\ }\bibfield  {title} {\enquote {\bibinfo
  {title} {Scaling laws for neural language models},}\ }\href@noop {}
  {\bibfield  {journal} {\bibinfo  {journal} {arXiv preprint arXiv:2001.08361}\
  } (\bibinfo {year} {2021})}\BibitemShut {NoStop}%
\bibitem [{\citenamefont {Li}\ \emph {et~al.}(2022)\citenamefont {Li},
  \citenamefont {Kovachki}, \citenamefont {Azizzadenesheli}, \citenamefont
  {Liu}, \citenamefont {Bhattacharya}, \citenamefont {Stuart},\ and\
  \citenamefont {Anandkumar}}]{li2022physics}%
  \BibitemOpen
  \bibfield  {author} {\bibinfo {author} {\bibfnamefont {Zongyi}\ \bibnamefont
  {Li}}, \bibinfo {author} {\bibfnamefont {Nikola}\ \bibnamefont {Kovachki}},
  \bibinfo {author} {\bibfnamefont {Kamyar}\ \bibnamefont {Azizzadenesheli}},
  \bibinfo {author} {\bibfnamefont {Burigede}\ \bibnamefont {Liu}}, \bibinfo
  {author} {\bibfnamefont {Kaushik}\ \bibnamefont {Bhattacharya}}, \bibinfo
  {author} {\bibfnamefont {Andrew}\ \bibnamefont {Stuart}}, \ and\ \bibinfo
  {author} {\bibfnamefont {Anima}\ \bibnamefont {Anandkumar}},\ }\bibfield
  {title} {\enquote {\bibinfo {title} {Physics-informed transformer for solving
  differential equations},}\ }\href@noop {} {\bibfield  {journal} {\bibinfo
  {journal} {arXiv preprint arXiv:2205.08853}\ } (\bibinfo {year}
  {2022})}\BibitemShut {NoStop}%
\bibitem [{\citenamefont {Breuer}\ \emph {et~al.}(2002)\citenamefont {Breuer},
  \citenamefont {Petruccione} \emph {et~al.}}]{breuer2002theory}%
  \BibitemOpen
  \bibfield  {author} {\bibinfo {author} {\bibfnamefont {Heinz-Peter}\
  \bibnamefont {Breuer}}, \bibinfo {author} {\bibfnamefont {Francesco}\
  \bibnamefont {Petruccione}},  \emph {et~al.},\ }\href@noop {} {\emph
  {\bibinfo {title} {The theory of open quantum systems}}}\ (\bibinfo
  {publisher} {Oxford University Press on Demand},\ \bibinfo {year}
  {2002})\BibitemShut {NoStop}%
\bibitem [{\citenamefont {Pearle}(2012)}]{pearle2012simple}%
  \BibitemOpen
  \bibfield  {author} {\bibinfo {author} {\bibfnamefont {Philip}\ \bibnamefont
  {Pearle}},\ }\bibfield  {title} {\enquote {\bibinfo {title} {Simple
  derivation of the lindblad equation},}\ }\href@noop {} {\bibfield  {journal}
  {\bibinfo  {journal} {European Journal of Physics}\ }\textbf {\bibinfo
  {volume} {33}},\ \bibinfo {pages} {805} (\bibinfo {year} {2012})}\BibitemShut
  {NoStop}%
\bibitem [{\citenamefont {Nielsen}\ and\ \citenamefont
  {Chuang}(2002)}]{nielsen2002quantum}%
  \BibitemOpen
  \bibfield  {author} {\bibinfo {author} {\bibfnamefont {Michael~A}\
  \bibnamefont {Nielsen}}\ and\ \bibinfo {author} {\bibfnamefont {Isaac}\
  \bibnamefont {Chuang}},\ }\href@noop {} {\enquote {\bibinfo {title} {Quantum
  computation and quantum information},}\ } (\bibinfo {year}
  {2002})\BibitemShut {NoStop}%
\bibitem [{\citenamefont {Johansson}\ \emph {et~al.}(2012)\citenamefont
  {Johansson}, \citenamefont {Nation},\ and\ \citenamefont
  {Nori}}]{johansson2012qutip}%
  \BibitemOpen
  \bibfield  {author} {\bibinfo {author} {\bibfnamefont {J~Robert}\
  \bibnamefont {Johansson}}, \bibinfo {author} {\bibfnamefont {Paul~D}\
  \bibnamefont {Nation}}, \ and\ \bibinfo {author} {\bibfnamefont {Franco}\
  \bibnamefont {Nori}},\ }\bibfield  {title} {\enquote {\bibinfo {title}
  {Qutip: An open-source python framework for the dynamics of open quantum
  systems},}\ }\href@noop {} {\bibfield  {journal} {\bibinfo  {journal}
  {Computer physics communications}\ }\textbf {\bibinfo {volume} {183}},\
  \bibinfo {pages} {1760--1772} (\bibinfo {year} {2012})}\BibitemShut {NoStop}%
\bibitem [{\citenamefont {Sachdev}(2011)}]{sachdev2011quantum}%
  \BibitemOpen
  \bibfield  {author} {\bibinfo {author} {\bibfnamefont {Subir}\ \bibnamefont
  {Sachdev}},\ }\href@noop {} {\emph {\bibinfo {title} {Quantum Phase
  Transitions}}}\ (\bibinfo  {publisher} {Cambridge University Press},\
  \bibinfo {year} {2011})\BibitemShut {NoStop}%
\bibitem [{\citenamefont {Islam}\ \emph {et~al.}(2011)\citenamefont {Islam},
  \citenamefont {Edwards}, \citenamefont {Kim}, \citenamefont {Korenblit},
  \citenamefont {Noh}, \citenamefont {Carmichael},\ and\ \citenamefont
  {Monroe}}]{islam2011onset}%
  \BibitemOpen
  \bibfield  {author} {\bibinfo {author} {\bibfnamefont {Rajibul}\ \bibnamefont
  {Islam}}, \bibinfo {author} {\bibfnamefont {Erica~E.}\ \bibnamefont
  {Edwards}}, \bibinfo {author} {\bibfnamefont {Kihwan}\ \bibnamefont {Kim}},
  \bibinfo {author} {\bibfnamefont {Simcha}\ \bibnamefont {Korenblit}},
  \bibinfo {author} {\bibfnamefont {Chibum}\ \bibnamefont {Noh}}, \bibinfo
  {author} {\bibfnamefont {H.}~\bibnamefont {Carmichael}}, \ and\ \bibinfo
  {author} {\bibfnamefont {Christopher}\ \bibnamefont {Monroe}},\ }\bibfield
  {title} {\enquote {\bibinfo {title} {Onset of a quantum phase transition with
  a trapped ion quantum simulator},}\ }\href@noop {} {\bibfield  {journal}
  {\bibinfo  {journal} {Nature Communications}\ }\textbf {\bibinfo {volume}
  {2}},\ \bibinfo {pages} {377} (\bibinfo {year} {2011})}\BibitemShut {NoStop}%
\bibitem [{\citenamefont {Zhang}\ \emph {et~al.}(2017)\citenamefont {Zhang},
  \citenamefont {Pagano}, \citenamefont {Hess}, \citenamefont {Kyprianidis},
  \citenamefont {Becker}, \citenamefont {Kaplan}, \citenamefont {Gorshkov},\
  and\ \citenamefont {Monroe}}]{zhang2017observation}%
  \BibitemOpen
  \bibfield  {author} {\bibinfo {author} {\bibfnamefont {Jiehang}\ \bibnamefont
  {Zhang}}, \bibinfo {author} {\bibfnamefont {Guido}\ \bibnamefont {Pagano}},
  \bibinfo {author} {\bibfnamefont {Paul~W.}\ \bibnamefont {Hess}}, \bibinfo
  {author} {\bibfnamefont {Antonis}\ \bibnamefont {Kyprianidis}}, \bibinfo
  {author} {\bibfnamefont {Peter}\ \bibnamefont {Becker}}, \bibinfo {author}
  {\bibfnamefont {Harvey}\ \bibnamefont {Kaplan}}, \bibinfo {author}
  {\bibfnamefont {Alexey~V.}\ \bibnamefont {Gorshkov}}, \ and\ \bibinfo
  {author} {\bibfnamefont {Christopher}\ \bibnamefont {Monroe}},\ }\bibfield
  {title} {\enquote {\bibinfo {title} {Observation of a many-body dynamical
  phase transition with a 53-qubit quantum simulator},}\ }\href@noop {}
  {\bibfield  {journal} {\bibinfo  {journal} {Nature}\ }\textbf {\bibinfo
  {volume} {551}},\ \bibinfo {pages} {601--604} (\bibinfo {year}
  {2017})}\BibitemShut {NoStop}%
\bibitem [{\citenamefont {Georgescu}(2020)}]{georgescu2020divincenzo}%
  \BibitemOpen
  \bibfield  {author} {\bibinfo {author} {\bibfnamefont {Iulia}\ \bibnamefont
  {Georgescu}},\ }\bibfield  {title} {\enquote {\bibinfo {title} {The
  divincenzo criteria 20 years on},}\ }\href@noop {} {\bibfield  {journal}
  {\bibinfo  {journal} {Nature Reviews Physics}\ }\textbf {\bibinfo {volume}
  {2}},\ \bibinfo {pages} {666--666} (\bibinfo {year} {2020})}\BibitemShut
  {NoStop}%
\bibitem [{\citenamefont {Pollard}\ and\ \citenamefont
  {Friesner}(1994)}]{pollard1994solution}%
  \BibitemOpen
  \bibfield  {author} {\bibinfo {author} {\bibfnamefont {W~Thomas}\
  \bibnamefont {Pollard}}\ and\ \bibinfo {author} {\bibfnamefont {Richard~A}\
  \bibnamefont {Friesner}},\ }\bibfield  {title} {\enquote {\bibinfo {title}
  {Solution of the redfield equation for the dissipative quantum dynamics of
  multilevel systems},}\ }\href@noop {} {\bibfield  {journal} {\bibinfo
  {journal} {The Journal of chemical physics}\ }\textbf {\bibinfo {volume}
  {100}},\ \bibinfo {pages} {5054--5065} (\bibinfo {year} {1994})}\BibitemShut
  {NoStop}%
\bibitem [{\citenamefont {Stone}(1948)}]{stone1948generalized}%
  \BibitemOpen
  \bibfield  {author} {\bibinfo {author} {\bibfnamefont {Marshall~H}\
  \bibnamefont {Stone}},\ }\bibfield  {title} {\enquote {\bibinfo {title} {The
  generalized weierstrass approximation theorem},}\ }\href@noop {} {\bibfield
  {journal} {\bibinfo  {journal} {Mathematics Magazine}\ }\textbf {\bibinfo
  {volume} {21}},\ \bibinfo {pages} {237--254} (\bibinfo {year}
  {1948})}\BibitemShut {NoStop}%
\bibitem [{\citenamefont {O'Donnell}\ and\ \citenamefont
  {Wright}(2016)}]{o2016efficient}%
  \BibitemOpen
  \bibfield  {author} {\bibinfo {author} {\bibfnamefont {Ryan}\ \bibnamefont
  {O'Donnell}}\ and\ \bibinfo {author} {\bibfnamefont {John}\ \bibnamefont
  {Wright}},\ }\bibfield  {title} {\enquote {\bibinfo {title} {Efficient
  quantum tomography},}\ }in\ \href@noop {} {\emph {\bibinfo {booktitle}
  {Proceedings of the forty-eighth annual ACM symposium on Theory of
  Computing}}}\ (\bibinfo {year} {2016})\ pp.\ \bibinfo {pages}
  {899--912}\BibitemShut {NoStop}%
\bibitem [{\citenamefont {Hsu}\ \emph {et~al.}(2024)\citenamefont {Hsu},
  \citenamefont {Kuo}, \citenamefont {Yu}, \citenamefont {Cai},\ and\
  \citenamefont {Hsieh}}]{hsu2024quantum}%
  \BibitemOpen
  \bibfield  {author} {\bibinfo {author} {\bibfnamefont {Ming-Chien}\
  \bibnamefont {Hsu}}, \bibinfo {author} {\bibfnamefont {En-Jui}\ \bibnamefont
  {Kuo}}, \bibinfo {author} {\bibfnamefont {Wei-Hsuan}\ \bibnamefont {Yu}},
  \bibinfo {author} {\bibfnamefont {Jian-Feng}\ \bibnamefont {Cai}}, \ and\
  \bibinfo {author} {\bibfnamefont {Min-Hsiu}\ \bibnamefont {Hsieh}},\
  }\bibfield  {title} {\enquote {\bibinfo {title} {Quantum state tomography via
  nonconvex riemannian gradient descent},}\ }\href@noop {} {\bibfield
  {journal} {\bibinfo  {journal} {Physical Review Letters}\ }\textbf {\bibinfo
  {volume} {132}},\ \bibinfo {pages} {240804} (\bibinfo {year}
  {2024})}\BibitemShut {NoStop}%
\end{thebibliography}%
